\documentclass[conference,compsoc]{IEEEtran} 

\usepackage[T1]{fontenc}
\usepackage[utf8]{inputenc}
\usepackage[english]{babel}
\usepackage{amsmath}
\usepackage{amsfonts}
\usepackage{amssymb}
\usepackage{mathtools}
\usepackage{braket}
\usepackage{amsthm}
\usepackage{pgfgantt}
\theoremstyle{definition}
\newtheorem{theorem}{Theorem}
\newtheorem{definition}{Definition}
\newtheorem{lemma}{Lemma}
\newtheorem{remark}{Remark}
\usepackage{enumitem}
\usepackage{caption}
\usepackage{listings}
\usepackage[pdftex=true,hyperindex=true,colorlinks=false,hidelinks=false]{hyperref}
\usepackage{tabularx}
\usepackage[lined,boxed,algonl]{algorithm2e}
\usepackage{etex}
\usepackage{tikz}
\usepackage{pgfplots}
\usepackage{placeins}
\usepackage{cite}
\newenvironment{aeq}{\begin{equation}
	\begin{aligned}
}{
	\end{aligned}
\end{equation}}
\numberwithin{equation}{section}
\numberwithin{theorem}{section}
\numberwithin{lemma}{section}
\numberwithin{definition}{section}
\definecolor{color1}{RGB}{0,0,90} % Color of the article title and sections
\definecolor{color2}{RGB}{0,20,20} % Color of the boxes behind the abstract and headings
\usepackage{hyperref} % Required for hyperlinks
\hypersetup{hidelinks,colorlinks,breaklinks=true,urlcolor=color2,citecolor=color1,linkcolor=color1,bookmarksopen=false,pdftitle={Title},pdfauthor={Author}}

\newcommand{\mf}[1]{\mathfrak{#1}}
\newcommand{\mbf}[1]{\mathbf{#1}}
\newcommand{\mc}[1]{\mathnormal{#1}}
\newcommand{\prop}{\text{prop}}
\newcommand{\auth}{\text{auth}}

\usepackage{cleveref}

\crefname{theorem}{Theorem}{Theorems}
\Crefname{theorem}{Theorem}{Theorems}
\crefname{thm}{Theorem}{Theorems}
\Crefname{thm}{Theorem}{Theorems}
\crefname{assump}{Assumption}{Assumptions}
\Crefname{assump}{Assumption}{Assumptions}
\crefname{problem}{Problem}{Problems}
\Crefname{problem}{Problem}{Problems}
\crefname{conjecture}{Conjecture}{Conjectures}
\Crefname{conjecture}{Conjecture}{Conjectures}
\crefname{proposition}{Proposition}{Propositions}
\Crefname{proposition}{Proposition}{Propositions}
\crefname{prop}{Proposition}{Propositions}
\Crefname{prop}{Proposition}{Propositions}
\crefname{cor}{Corollary}{Corollaries}
\Crefname{cor}{Corollary}{Corollaries}
\crefname{lem}{Lemma}{Lemmas}
\Crefname{lem}{Lemma}{Lemmas}
\theoremstyle{definition}
\crefname{definition}{Definition}{Definitions}
\Crefname{definition}{Definition}{Definitions}
\crefname{defn}{definition}{definitions}
\Crefname{defn}{Definition}{Definitions}
\crefname{remark}{Remark}{Remarks}
\Crefname{remark}{Remark}{Remarks}
\crefname{rmk}{Remark}{Remarks}
\Crefname{rmk}{Remark}{Remarks}
\crefname{example}{Example}{Examples}
\Crefname{example}{Example}{Examples}
\crefname{table}{Table}{Tables}
\Crefname{table}{Table}{Tables}
\crefname{align}{}{}
\Crefname{align}{}{}
\crefname{equation}{eq.}{eqs.}
\Crefname{equation}{Eq.}{Eqs.}
\crefname{step}{Step}{Steps}
\Crefname{step}{Step}{Steps}
\crefname{protocol}{protocol}{protocols}
\Crefname{protocol}{Protocol}{Protocols}
\crefname{algorithm}{algorithm}{algorithms}
\Crefname{algorithm}{Algorithm}{Algorithms}

\title{Composable Security of Distributed Symmetric Key Establishment Protocol} % Article title
\author{
    \IEEEauthorblockN{
        Jie Lin\IEEEauthorrefmark{1}\IEEEauthorrefmark{2},
        Manfred von Willich\IEEEauthorrefmark{1},
        Hoi-Kwong Lo\IEEEauthorrefmark{1}\IEEEauthorrefmark{2}
    }
    \IEEEauthorblockA{
        \IEEEauthorrefmark{1}Quantum Bridge Technologies Inc., 108 College St., Toronto, ON, Canada
    }
    \IEEEauthorblockA{
        \IEEEauthorrefmark{2}Department of Electrical and Computer Engineering, University of Toronto, 10 King’s College Road, Toronto, ON, Canada 
    }
}

\begin{document}

%------------------------------------------------------------------------------------------------------------------------------------

%	ABSTRACT
%------------------------------------------------------------------------------------------------------------------------------------

%------------------------------------------------------------------------------------------------------------------------------------

%\flushbottom % Makes all text pages the same height

\maketitle % Print the title and abstract box

\begin{abstract}
The Distributed Symmetric Key Establishment (DSKE) protocol provides secure secret exchange (e.g., for key exchange) between two honest parties that need not have had prior contact, and use intermediaries with whom they each securely share confidential data. We show the composable security of the DSKE protocol in the constructive cryptography framework of Maurer. Specifically, we prove the security (correctness and confidentiality) and robustness of this protocol against any computationally unbounded adversary, who additionally may have fully compromised a bounded number of the intermediaries and can eavesdrop on all communication. As DSKE is highly scalable in a network setting with no distance limit, it is expected to be a cost-effective quantum-safe cryptographic solution to safeguarding the network security against the threat of quantum computers. 
\end{abstract}

\tableofcontents % Print the contents section

%\thispagestyle{empty} % Removes page numbering from the first page

%------------------------------------------------------------------------------------------------------------------------------------
\pagestyle{plain}
\section{Introduction} \label{sec:intro}

Public key infrastructure has played a key role in today's network security. As tremendous experimental progress has been made in quantum computing in the last three decades, the quantum threat to communication security is widely recognized by governments, industries and academia. To counter the quantum threat to public key infrastructure, there are three major categories of solutions: post-quantum cryptography (PQC) \cite{Bernstein2017}, quantum key distribution (QKD) \cite{Xu2020} and pre-shared keys (PSKs). An advantage of PQC is that it is software-based and can be implemented in the Internet without dedicated special hardware. However, since PQC is based on unproven computational intractability assumptions, the risk of an unexpected security breach of PQC is high \cite{Geer2023nist, Townsend2023}. While QKD provides information-theoretic security, the cost of QKD can be quite high and there are often limits to its key rate and distance. Without quantum repeaters, QKD is not yet a scalable solution in the global Internet. PSK has the advantage of being quantum-safe because it either employs one-time-pad or symmetric key crypto-systems that, unlike public key crypto-systems, are resistant to quantum attacks. Unfortunately, up till now, PSK has the disadvantage of being unscalable in a network of many users. This is because each user has to share a key with another user to communicate securely. Therefore, with a large number, say $N$, users, there are $N(N-1)/2$ pairs of users. Each time when a new user joins the system, the existing users need to share a new key with the new user in order to communicate with them. This is highly inconvenient and costly.

Recently, the distributed symmetric key establishment (DSKE) protocol \cite{Lo2022, DSKEpatent} has been introduced to provide a scalable solution that is also information-theoretically secure. DSKE has three major advantages.
First, DSKE is highly scalable. With DSKE, when a new user comes on-line, no new key materials need to be delivered to the existing users, unlike PSKs.
Second, DSKE distributes trust among multiple third parties, which are called Security Hubs in DSKE. Provided that the number of compromised Security Hubs (those that deviate from honest behaviours defined in the protocol) is below a predetermined threshold, DSKE provides {\it information-theoretic} security.
This implies that DSKE avoids any single points of failure.
Third, unlike QKD, DSKE has no distance limit and does not require dedicated optical fibres.  

DSKE is also highly versatile and can be used in many applications including mobile phones, network security and embedded systems such as Internet of Things devices. More concretely, the DSKE protocol may be used to agree sequences of data between multiple parties with a number of security properties, including quality of randomness, robustness and information-theoretically secure authenticity and confidentiality, to support a large range of use cases.
As an example, DSKE allows two participants, Alice and Bob, to agree on an information-theoretically secure shared symmetric key, which can then be used for encrypting some data for Alice to send to Bob. Data integrity can also be achieved by the DSKE protocol through authentication by message tag with information-theoretically secure authentication.

In this paper, we provide a rigorous security proof for the DSKE protocol. To allow other cryptographic applications to use the sequence of data agreed by honest parties from the DSKE protocol (which we call the \textit{secret} for the remainder of this paper) in a secure way, it is important to prove the composable security of the DSKE protocol. To do so, we show the DSKE protocol is $\epsilon$-secure in the constructive cryptography framework by Maurer \cite{Maurer2011}, which implies universal composability \cite{canetti2001universally,Canetti2020}. We first prove the security of the DSKE protocol by assuming the availability of perfectly authenticated channels. Then we use the composability theorem in the constructive cryptography framework to replace perfectly authenticated channels by a practical authentication protocol and insecure communication channels. 
We also show the DSKE protocol is $\epsilon$-robust, which means with a probability at least $1-\epsilon$, the protocol does not abort if an adversary is passive on communication channels (see \Cref{sec:threat_model} for a precise definition of being passive).

Another novelty of this work is that our scheme to verify the correctness of the secret does not rely on pre-shared keys between the communicating parties. A typical message authentication scheme requires a pre-shared key to securely verify the message's authenticity. In our scheme, we transmit the key together with the message, relying only on the same security assumptions that are needed for confidentiality. We show the correctness and confidentiality of such a secret validation scheme.

\subsection{Structure of this document}
In \Cref{sec:preliminary}, we present technical preliminaries that are relevant for our security proof. In particular, we state the results about our choice of 2-universal hash function family. We also review the constructive cryptography framework \cite{Maurer2011}. We provide a short summary of the DSKE protocol in \Cref{subsec:skeleton_protocol_description} and direct readers to \Cref{app:protocol} for a detailed description of the general DSKE protocol with simplifying protocol parameter choices. A table of symbols used in the protocol description is given in \Cref{table:symbols}. Readers seeking greater familiarity with the DSKE protocol may read \cite{Lo2022}. 

In \Cref{sec:security}, we discuss the security definition for the DSKE protocol and prove the security of a variant of the protocol, which is the general DSKE protocol under the assumption that perfectly authenticated channels are freely available. 

In \Cref{sec:proof}, we then prove the security of the general DSKE protocol.

We show results about the security of hashing for message authentication in \Cref{app:hash_message} and about the security of Shamir secret sharing with validation in \Cref{app:hash_secret}. 

\FloatBarrier

\section{Technical Preliminaries} \label{sec:preliminary}

\subsection{Notation}\label{subsec:notation}

We define some common notations used throughout this paper to assist our discussions. In particular, $A$ refers to Alice, $B$ refers to Bob, $E$ refers to Eve, and $P_i$ refers to the unique identifier assigned to the Security Hub indexed by $i\in\{1,\dots,n\}$. We highlight that $F$ denotes the same finite field throughout, $|F|$ is the number of elements in $F$, and by \textit{element} we mean an element of $F$. By \textit{length}, we mean \textit{number of field elements}. 

As we often need to write a collection of symbols, (\dots, $Y_i$, \dots), where the index $i$ iterates through each element of a list $S$, we define a shorthand notation $(Y_i)_{i \in S}$ or simply $(Y_i)_{i}$ when the list $S$ is clear from the context (e.g., we write $(Y_1,Y_2,Y_3)$ as $(Y_i)_{i \in [1,2,3]}$ or just $(Y_i)_{i}$). We use capital letters (e.g., $Y$) to denote random variables and lowercase letters (e.g., $y$) to denote a particular value of the random variable. We use \textit{uniform} to mean \textit{uniformly distributed}. We use $\delta_{x,y}$ to denote the Kronecker delta function. Between two sequences, $\parallel$ denotes the concatenation operation. Between two resources, $\parallel$ denotes the parallel composition operation.

Throughout, given a finite sample space $\Omega$ and a probability distribution $p$ over it, we construct the corresponding probability space $(\Omega, \mathcal{F}, \mbf{P})$ by setting the event space $\mathcal{F}$ as the collection of all subsets of $\Omega$, and $\mbf{P}(A)=\sum_{\omega\in A} p(\omega)$.

\Cref{table:symbols} in \Cref{app:protocol} lists symbols related to the protocol description. Other notations are introduced where they first appear.

\subsection{Mathematical fundamentals}

\begin{enumerate}[leftmargin=0cm]

    \item[] \textbf{Fields:} Let $F=\mathrm{GF}(p^{r})$, a Galois field with $p$ a prime number and $r>0$ an integer. $F$ is a finite field with $|F|=p^{r}$ elements. Multiplication by a nonzero element of $F$ is a bijective mapping $F\to F$.

    \item[] \textbf{Addition of vectors:} The additive group of the vectors of dimension $m$ over $\mathrm{GF}(p^{r})$, namely $\mathrm{GF}(p^{r})^m$, is the same as the additive group of $\mathrm{GF}(p)^{rm}$.  This allows addition of vectors of suitable lengths over different fields to be compatible (i.e., to be the same group), when both fields have the same characteristic $p$.  For example, we can use this to match a Shamir secret sharing scheme of vectors over $\mathrm{GF}(2^{8})$, with a hash function over the field $\mathrm{GF}(2^{128})$ using bit-wise exclusive-or as addition while maintaining the properties that depend on this addition.  For simplicity of exposition, we use the same field for both the secret sharing scheme and the validating hash function.

    \item[] \textbf{Probability theory:}

    For any proposition $\prop(X, Y)$ of random variables $X$ and $Y$,%
    \begin{aeq}\label{eq:prob_prop}
        & \Pr(\prop(X, Y)) \\
        & = \sum_{y\in F} \Pr(\prop(X, Y) | Y = y) \Pr(Y = y).
    \end{aeq}

    When summing over all values of a random variable $Y \in \Omega$, where $\Omega$ is the relevant sample space for $Y$,%
    \begin{aeq}\label{eq:prob_sum_to_one}
        \sum_{y \in \Omega} \Pr(Y = y) = 1. 	
    \end{aeq}%
    When $X,Y\in \Omega$ are mutually independent and $X$ uniform,%
    \begin{aeq}\label{eq:prob_uniform}
        \Pr(X = x | Y = y) = 1/|\Omega|. 	
    \end{aeq}%
    Random variables $X$ and $Y$ are mutually independent if%
    \begin{aeq}\label{eq:prob_independent}
        \forall (x, y): & \Pr( (X, Y) = (x, y) ) \\
        & = \Pr( X = x ) \Pr( Y = y ). 	
    \end{aeq}

    \item[] \textbf{Statistical distance:} For two probability distributions $P_X$ and $Q_X$ of a random variable $X$ that can take any value in some set $\mathcal{X}$, the statistical or total variation distance between $P_X$ and $Q_X$ has the following properties:%
    \begin{aeq}\label{eq:statistical_distance_property}
        &\frac{1}{2}\sum_{x \in \mathcal{X}}|P_X(x)-Q_X(x)| \\
         = &\max_{\mathcal{X}':\mathcal{X}'\subseteq \mathcal{X}}\sum_{x \in \mathcal{X}'}(P_X(x)-Q_X(x)) \\
         = &\sum_{x: P_X(x) \geq Q_X(x)} (P_X(x)-Q_X(x)).
    \end{aeq}

\end{enumerate}

\subsection{Cryptographic primitives}

\subsubsection{Message validation (authentication)}

The following \cref{thm:hash_function_correctness} is a consequence of the hashing scheme being a Carter--Wegman universal hash function family \cite{wegman1981new}.  This polynomial function family is described by Bernstein \cite{Bernstein2007}.

\begin{theorem}\label{thm:hash_function_correctness}
    Denote $\mbf{v} = (v_{1},...,v_{s})$ and $\mbf{v}^{*} = (v^{*}_{1},...,v^{*}_{s})$. \\
    Let $\Omega=F^2$ be a sample space with uniform probability. \\
    Let $h_{C,D}(\mbf{v})=d+\sum_{j=1}^{s}c^{j} v_{j}$ define a family of functions with random variables $(C,D)\in\Omega$ as selection parameters.
    Let $s\ne0$. Let $t\in F$ be given.
    Then, $\max_{t^{*},\mbf{v}^{*}\ne\mbf{v}} \Pr(t^{*}=h_{C,D}(\mbf{v}^{*})~|~t=h_{C,D}(\mbf{v})) = \min(\frac{s}{|F|},1)$.
\end{theorem}
\begin{proof}
	See \Cref{thm:message_correctness} in \Cref{app:hash_message}.
\end{proof}

Given a message $\mbf{v}$ together with a valid tag $t=h_{c,d}(\mbf{v})$, with $(c,d)$ being uniform (described by the random variables $C,D$), \cref{thm:hash_function_correctness} tells us that the maximal success probability for a forger to create a differing message $\mbf{v}^{*}$ and a new hash value $t^{*}$ such that these correspond (i.e., that $t^{*}=h_{C,D}(\mbf{v}^{*})$) is at most $\min(\frac{s}{|F|},1)$ when the message length $s$ is nonzero.
In the other words, (for large $|F|$) not knowing $(c,d)$ makes it highly improbable for the forger to successfully substitute a tag and \textit{differing} message.

The theorem does not apply without a prior message. In the absence of a prior message, the probability that a forged message will be validated is $\frac{1}{|F|}$, which is independent of the length $s$.

\subsubsection{Secret confidentiality and validation}

Shamir introduced a secret sharing scheme that produces $n$ shares, any $k$ of which are sufficient to reconstruct the secret, but any $k-1$ of which give no information about the secret \cite{shamir1979share}.

\begin{theorem}\label{thm:secret_confidentiality}
    In a Shamir secret sharing scheme with threshold $k$, the shared secret is independent of any subset of the shares of size $k-1$ or less.
\end{theorem}
\begin{proof}
    See \cref{thm:shamir_confidentiality} in \Cref{app:hash_secret}.
\end{proof}

\Cref{thm:secret_confidentiality} tells us that, under the constraint of access to only $k-1$ shares, a Shamir threshold-$(n,k)$ secret sharing scheme has perfectly secure confidentiality. 

\begin{theorem}\label{thm:secret_linearity}
    In a Shamir secret sharing scheme with threshold $k$, for any given set of $k$ shares, the secret is a linear combination of the shares.
\end{theorem}
\begin{proof}
    See \cref{thm:shamir_linearity} in \Cref{app:hash_secret}.
\end{proof}

\begin{theorem}\label{thm:hash_function_confidentiality}
    Denote $\mbf{Y}=(Y_{(1)},\dots,Y_{(m)})$. \\
    Let $\Omega=F^{3+m}$ be a sample space.
    Let $(C,D,E,\mbf{Y})\in\Omega$, with $D$ uniform and independent of $C,E$ and $\mbf{Y}$. \\
    Let $T=h'_{C,D,E}(\mbf{Y})=D+CE+\sum_{j=1}^{m}C^{j+1}Y_{(j)}$. \\
    Then $T$ is independent of $\mbf{Y}$.
\end{theorem}
\begin{proof}
	This follows directly from $D$ being uniform and independent of $X$ with $T=D+X$, where $X=CE+\sum_{j=1}^{m}C^{j+1}Y_{(j)}$, due to \Cref{lem:indepsum} in \Cref{app:hash_secret}.
\end{proof}

\Cref{thm:hash_function_confidentiality} tells us that no information is obtained from the tag ($o^A$ as $T$) about the secret ($S^A$ as $\mbf{Y}$).  Thus, given the perfectly secure confidentiality provided by a Shamir secret sharing scheme for up to $k-1$ shares being known, publishing the tag does not impact this confidentiality.

Given a Shamir secret sharing scheme, the resulting secret is malleable: a modification of any share will modify the reconstructed secret in a known way, in the same way that modifying the ciphertext of one-time-pad encryption modifies the decrypted plaintext.  Normally, a message authentication code would be employed to allow detection of such a change, but this needs a shared key to implement.  Transmitting a validation key via the same secret sharing scheme violates the normal premise for authentication: that the validation key is assured to be the same at both sides.  \Cref{thm:secret_validation} gives us a way to transmit such a key using the same (malleable) secret sharing scheme while providing authenticity, under the same premise that the secret sharing scheme already has.  We believe that this is a novel construction \footnote{A related objective is presented in \cite{ostrev2019composable}.}.

\begin{theorem}\label{thm:secret_validation}
    Denote $\mbf{y}=(y_{(1)},\dots,y_{(m)})$ and $\mbf{y}'=(y'_{(1)},\dots,y'_{(m)})$.
    Let $\Omega=F^3$ be a sample space with uniform probability.
    Let $h'_{C,D,E}(\mbf{y})=D+CE+\sum_{j=1}^{m}C^{j+1}y_{(j)}$ define a family of hash functions with random variables $(C,D,E)\in\Omega$ as selection parameters. Let $m\ne0.$
    Then, $\max_{t',c',d',e',\mbf{y}'\ne\mbf{y}} \Pr(t+t'=h'_{C+c',D+d',E+e'}(\mbf{y}+\mbf{y}')~|~t=h'_{C,D,E}(\mbf{y}))\le\min(\frac{m+1}{|F|},1)$.
    \end{theorem}
\begin{proof}
	See \cref{thm:shamir_correctness} in \Cref{app:hash_secret}.
\end{proof}

\Cref{thm:secret_validation} gives us an upper bound on the probability of validating a secret with a nontrivial alteration, with three confidential uniform elements and the alteration constrained to \textit{addition} of a vector to $(C,D,E,\mbf{y})$. The premise for the theorem is assured by the secret sharing scheme. There are $n\choose k$ subsets of $k$ shares and hence reconstructions, but a probability is upper-bounded at $1$, giving an upper bound on the probability of any of the ${n\choose k}$ reconstructed secret candidates passing the validation of $\epsilon\le\min({n\choose k}\frac{m+1}{|F|},1)$.

Combining \Cref{thm:secret_confidentiality,thm:secret_linearity,thm:hash_function_confidentiality,thm:secret_validation} leads to the following result:

The $(n,k)$-Shamir secret sharing scheme and a polynomial validation code for transmitting $3+m$ elements of the same finite field $F$, where the first 3 elements are uniform, of which the first $3$ are consumed as $u^A$ and the remaining $m$ elements are $S^A$, has
\begin{itemize}
    \setlength\itemsep{0em}
    \item $\epsilon$-secure correctness with $\epsilon\le\min({n\choose k}\frac{m+1}{|F|},1)$
    \item perfectly secure confidentiality
\end{itemize}
against a computationally unbounded adversary who can access and modify up to $k-1$ shares and block but not access or modify any of the remaining shares.

This leverages the secrecy provided by the sharing scheme to deliver the key used for secure validation of the secret that is simultaneously delivered.  This is novel in the sense that it ensures correctness (authentication) in addition to retaining the confidentially of the secret sharing scheme without imposing additional security requirements, i.e., it assumes only that no more than $k-1$ of the shares are compromised, as with the sharing scheme, whereas normally the key for validation would be assumed to be pre-shared.

\subsection{Composability and constructive cryptography}\label{subsec:abstrac_crypto}

As a cryptographic protocol is often combined with many other protocols, it is important to prove the security of the protocol in a composable security framework. The composability result in such a framework asserts that in analyzing the security of a complex protocol, one can simply decompose it into various subprotocols and analyze the security of each. Provided that each real subsystem constructed by a subprotocol is close to an ideal subsystem within some $\epsilon$, which is quantified by some distance measure (a pseudo-metric in abstract cryptography), the real system constructed from the combined protocol will then be close to the combined ideal system. The sum of the $\epsilon$-values for the subprotocols gives an $\epsilon$-value for the combined protocol.

The abstract cryptography framework \cite{MaurerRenner2011} uses a top-down approach. In this framework, one states the definitions and proves the security from the highest possible level of abstraction and proceeds downwards. When one deals with a specific instantiation of an abstract system in a lower level, as long as the lower level system satisfies the properties required in the higher-level proof, the composed protocol is secure. 

The constructive cryptography \cite{Maurer2011} is an application of the abstract cryptography framework to defining classical cryptographic primitives in a constructive way. In this framework, one specifies the resources that are used by a protocol, their required properties and some desired functionalities of an ideal system. If a protocol constructs the ideal system from the given resources, then it is secure. A protocol that is secure in this sense is also universally composable secure. We briefly review main concepts and terminologies from the constructive cryptography below and refer to \cite{Maurer2011} for further details.

\subsubsection{Resource}

An $\mc{I}$-resource is a system with interface label set $\mc{I}$ (e.g. $\mc{I} = \{A, B, E\}$). Resources can be composed together via the parallel composition operation. For two resources $\mf{R}_1$ and $\mf{R}_2$, we write $\mf{R}_1 \parallel \mf{R}_2$ as the resource after the parallel composition. 

On the set of resource systems, one can define a pseudo-metric to quantify the closeness between any two resources. We present the definition of pseudo-metric.
 
\begin{definition}[Pseudo-metric]
 	A function $d: \Omega \times \Omega \to \mathbb{R}^{\ge0}$ is a pseudo-metric on the set $\Omega$ if for all $a, b, c \in \Omega$, $d(a,a) =0, d(a, b) = d(b, a), d(a, b) \leq d(a,c) + d(c, b)$. 
\end{definition}

If the pseudo-metric also satisfies $d(a, b) = 0 \implies a=b$, then it is a metric.

\subsubsection{Converter}

A converter system, usually denoted by a Greek letter (e.g. $\pi$) is a system with two interfaces, an inside interface and an outside interface, which transforms one resource into another. The inside interface of a converter can be connected to an interface of a resource, and the outside interface becomes the new interface of the constructed resource. Two converters can be composed together via either serial or parallel composition operator. For two converters $\alpha$ and $\beta$, we write $\alpha \beta$ as the converter formed by serial composition and $\alpha \parallel \beta$ as the converter formed by parallel composition.

A protocol $\pi = \{\pi_i\}_{i \in \mc{J}}$ is a set of converters $\pi_i$ with $\mc{J} \subseteq \mc{I}$.

\subsubsection{Distinguisher and distinguishing advantage}

For $n$-interface resources, a distinguisher $\mf{D}$ is a system with $n+1$ interfaces, where $n$ interfaces connect to the interfaces of a resource and the other interface outputs a bit. In the constructive cryptography framework (as in many other composable security frameworks), the real and ideal systems are interactive black boxes that are given to the distinguisher with equal probability. Its task is to guess which system is in the black box. The output bit indicates its guess.

For a class of distinguishers $\mathbb{D}$, the distinguishing advantage for two resources $\mf{R}$ and $\mf{S}$ is%
\begin{aeq}\label{eq:distinguishing_advantage}
    d(\mf{R}, \mf{S}) := \max_{\mf{D} \in \mathbb{D}} (\Pr[\mf{D}\mf{R} = 1] - \Pr[\mf{D}\mf{S}=1]),
\end{aeq}%
where $\mf{D}\mf{R}$ is the binary random variable corresponding to $\mf{D}$ connected to $\mf{R}$, and $\mf{D}\mf{S}$ is defined similarly. Each binary random variable may take the value $1$ if the distinguisher $\mf{D}$ guesses the resource that it is connected to is the ideal resource system, and take the value $0$ if it guesses the real resource system\footnote{Exchanging 0 and 1 does not impact on the distinguishing advantage.}.  

As we are interested in information-theoretic security, the class of distinguishers $\mathbb{D}$ is the set of all computationally unbounded distinguishers, which may have access to quantum computers. In particular, for the DSKE protocol, we consider only classical systems. For two classical systems $\mf{R}$ and $\mf{S}$, once we show that they are indistinguishable for the set of all classical distinguishers, then they are automatically indistinguishable for the set of all quantum distinguishers since classical computers can simulate quantum ones (even if the simulation might be inefficient) \cite[Remark 3]{Portmann2014}. Thus, we consider only the set of all possible classical distinguishers in this paper.

We summarize a few useful properties of the distinguishing advantage. The distinguishing advantage defined in \Cref{eq:distinguishing_advantage} is a pseudo-metric on the set of resources. It is non-increasing under (serial or parallel) composition of any two systems, that is, for any resource systems $\mf{R}, \mf{S}, \mf{T}$, and any converter $\alpha$ (with $\alpha^i$ denoting $\alpha$ converting interface $i$),%
\begin{aeq}\label{eq:metric_serial_composition}
    d(\alpha^i\mf{R}, \alpha^i\mf{S})& \leq d(\mf{R}, \mf{S}),\\
\end{aeq}%
and%
\begin{aeq}\label{eq:metric_parallel_composition}
    d(\mf{R} \parallel \mf{T}, \mf{S} \parallel \mf{T}) &\leq d(\mf{R}, \mf{S}),\\
    d(\mf{T} \parallel \mf{R}, \mf{T} \parallel \mf{S}) &\leq d(\mf{R}, \mf{S}).\\
\end{aeq}%
 
\subsubsection{Composable security definition}
 
\begin{definition}\label{def:composable_security}
    Let $\mf{R}$ and $\mf{S}$ be resources with the interface label set $\mc{I} =\{A, B, E\}$. We say that a protocol $\pi = (\pi^A, \pi^B)$ securely constructs $\mf{S}$ out of $\mf{R}$ within $\epsilon$ if the following conditions hold:
    \begin{enumerate}[label=(\roman*)]
        \item For converters $\alpha^E$ and $\gamma^E$ that emulate an honest behaviour at the $E$-interface of each system and block any distinguisher from the access of $E$-interface of each system,%
        \begin{aeq}
            d(\pi\mf{R}\alpha^E, \mf{S}\gamma^E) \leq \epsilon.
		\end{aeq}%
		\item There exists a converter $\sigma^E$ such that%
		\begin{aeq}
			d(\pi \mf{R}, \mf{S}\sigma^E) \leq \epsilon.
		\end{aeq}%
	\end{enumerate}
    We use the construction notation $\mf{R} \xrightarrow{(\pi, \epsilon)} \mf{S}$ for this case.
\end{definition}
\begin{remark}
    The first condition in \cref{def:composable_security} captures the correctness of the protocol when no adversary is present. In this case, the adversarial controls covered by $\alpha^E$ and $\gamma^E$ are not accessible to the distinguisher. The second condition captures the situation when an adversary is present. In that case, the converters $\alpha^E$ and $\gamma^E$ are removed so that the distinguisher has full access to Eve's interfaces, which are the $E$-interface in the real system, and the $E$-interface of a simulator $\sigma^E$ that is attached to the ideal system. 
\end{remark}

\subsubsection{Composability theorem}
 
\begin{theorem}[{\cite[Theorem 1]{Maurer2011}}]\label{thm:composability}
	The security definition in \cref{def:composable_security} is generally composable if the pseudo-metric $d$ is compatible with the cryptographic algebra\footnote{We refer to \cite{Maurer2011} for the definitions of cryptographic algebra and compatibility of the pseudo-metric with cryptographic algebra. The pseudo-metric used in this paper is compatible with the underlying cryptographic algebra since we consider information-theoretic security so that if a distinguisher in the set of all distinguishers $\mathbb{D}$ is composed with another arbitrary system, it is still in $\mathbb{D}$. To avoid further distraction, we omit this definition here.}. The following statements hold:
	\begin{enumerate}[label=(\roman*)]
		\item If a protocol $\pi$ securely constructs a system $\mf{S}$ out of a system $\mf{R}$ within $\epsilon$ and another protocol $\pi'$ securely constructs a system $\mf{T}$ out of a system $\mf{S}$ within $\epsilon'$, then the serial composition $\pi\pi'$(running the protocol $\pi'$ after the protocol $\pi$) securely constructs the system $\mf{T}$ out of the system $\mf{R}$ within $\epsilon+\epsilon'$, i.e.,%
        \begin{aeq}
			& (\mf{R} \xrightarrow{(\pi, \epsilon)} \mf{S}) \;\land\; (\mf{S} \xrightarrow{(\pi', \epsilon')} \mf{T}) \\
            \implies & \mf{R} \xrightarrow{(\pi\pi', \epsilon+\epsilon')} \mf{T};
		\end{aeq}%
		\item If a protocol $\pi$ securely constructs a system $\mf{S}$ out of a system $\mf{R}$ within $\epsilon$ and another protocol $\pi'$ securely constructs a system $\mf{S}'$ out of a system $\mf{R}'$ within $\epsilon'$, then the parallel composition $\pi \parallel \pi'$ securely constructs the system $\mf{S} \parallel \mf{S}'$ out of the system $\mf{R} \parallel \mf{R}'$ within $\epsilon+\epsilon'$, i.e.,%
        \begin{aeq}
			& (\mf{R} \xrightarrow{(\pi, \epsilon)} \mf{S}) \;\land\; (\mf{R}' \xrightarrow{(\pi', \epsilon')} \mf{S}') \\
            \implies & \mf{R}\parallel\mf{R}' \xrightarrow{(\pi\parallel\pi', \epsilon+\epsilon')} \mf{S} \parallel \mf{S}';
        \end{aeq}%
        \item When a trivial converter, which applies the identity transformation to the resource that it connects to, is applied to a system $\mf{R}$, it perfectly constructs the system $\mf{R}$ out of itself, i.e.,%
        \begin{aeq}
            \mf{R} \xrightarrow{(\mathbf{1}, 0)} \mf{R},
        \end{aeq}%
        where $\mathbf{1}$ denotes the trivial converter.
	\end{enumerate}
\end{theorem}

\subsection{Synopsis of DSKE protocol}\label{subsec:skeleton_protocol_description}
We give a synopsis of the DSKE protocol without technical details to assist a high-level understanding of our security proof. We refer to \Cref{app:protocol} for a detailed description of the protocol steps. 

\begin{figure}[t]
	\includegraphics[width=\linewidth]{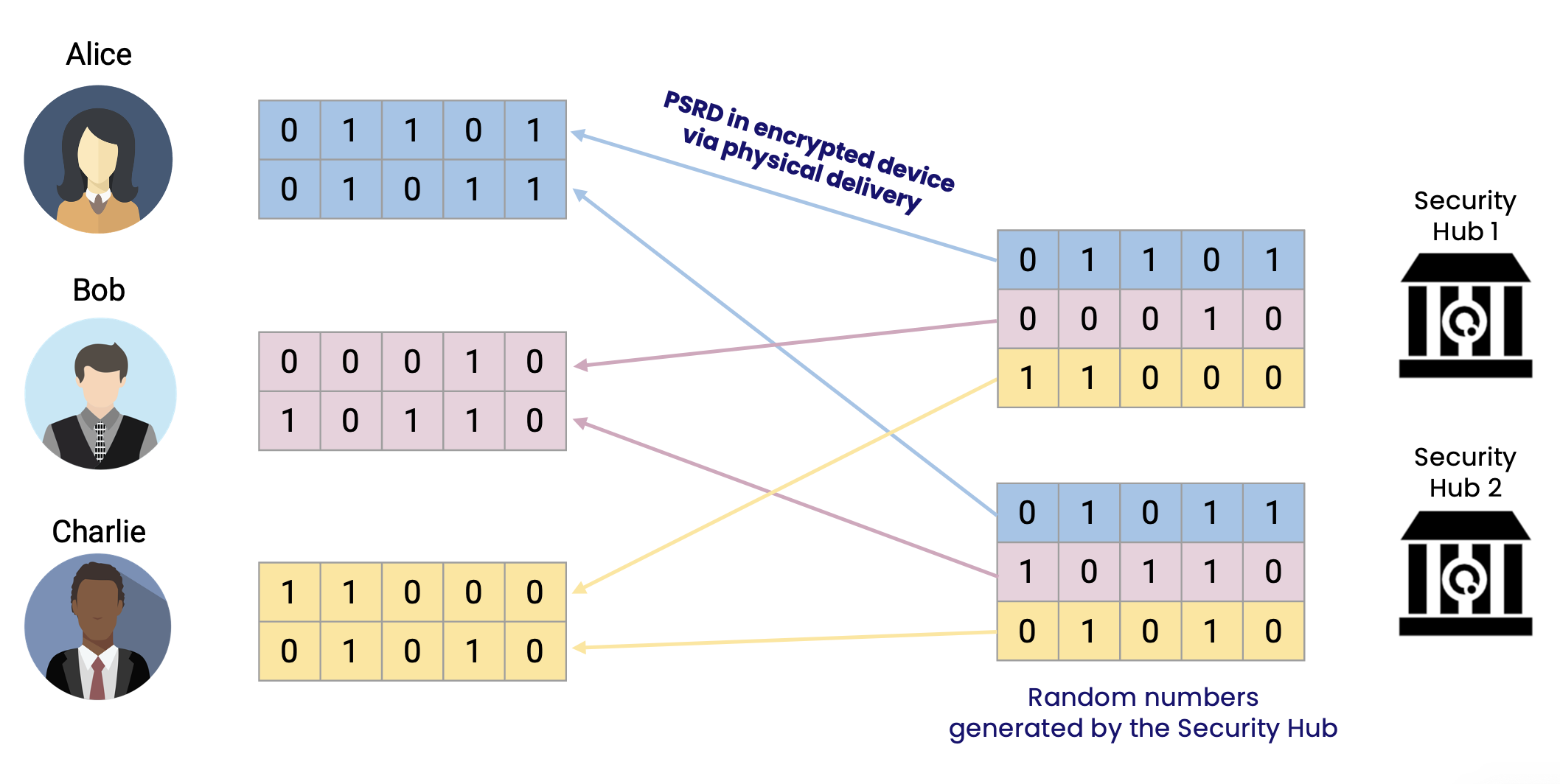}
	\caption{(Modified from Figure 1 of \cite{Lo2022}) The results of the one-time set-up: Steps 1 (\textit{PSRD generation and distribution}) and 2 (\textit{Peer identity establishment}) of the protocol. DSKE users Alice, Bob and Charlie share an ordered table of PSRD with each of the Security Hubs. Each Security Hub only knows its own part of the users' tables. For this illustration only, the PSRD is shown as bits.}
	\label{fig:protocol}
\end{figure}

We work with a two-user key agreement protocol in a network setting with a large number, $N$, of potential ``end users'' in the presence of a number, say $n$, of third parties called Security Hubs. The Security Hubs are numbered from $1$ to $n$ and an identifier $P_i$ is assigned to the $i$th Hub. The main goals are to guarantee that both users agree on the same secret and to protect the privacy of the agreed secret from potential adversaries, including other end users and the Security Hubs. During the one-time set-up (Steps (1) and (2) of the protocol as described in \Cref{app:protocol}), secure channels are assumed between the end users and Security Hubs. Those secure channels enable the end users and the Security Hubs to share some pre-shared random data (PSRD). Once the one-time setup has been done, we are in the situation described in \Cref{fig:protocol}. \Cref{fig:protocol} shows an example of a network with the users such as Alice, Bob, and Charlie together with two Security Hubs. Each user shares a table of PSRD with each Security Hub. Note that each Security Hub knows only the values of their own part of the shared random data, but has no information on the values of the shared random data of other Security Hubs.

We assume one-way communication from Alice to Bob, that is, Alice requests via the Security Hubs to exchange a secret with Bob. Alice is not interested in receiving information from Bob at all during the execution of the protocol, and may have only unidirectional communication available. 

\begin{remark}
    Note that if Alice and Bob share a string of random secret bits for use as a one-time-pad, it is important that Alice and Bob do not use the same sequence of values simultaneously for encryption. If they were to do so, when Alice uses a key bit, $k_i$ to encrypt a message, $a_i$, then she is sending $c_i = k_i \oplus a_i$ to a communication channel, where $\oplus$ denotes bitwise exclusive-or, and should Bob (inadvertently) uses the same key, $k_i$, to encrypt another message, $b_i$, then he is sending $d_i = k_i \oplus b_i$. Then, an eavesdropper who possesses both $c_i$ and $d_i$ can compute the parity of the two bits, $c_i \oplus d_i$, to obtain $c_i \oplus  d_i = (k_i \oplus a_i) \oplus (k_i \oplus b_i) = a_i \oplus b_i$, thus recovering the parity of the pair $(a_i, b_i)$. This would be a serious security breach. To avoid the above problem, it will be simplest to pre-assign each random string to a particular sender so that only one party is allowed to be the sender of the communication using the particular sequence as a one-time-pad key or one-time authentication key.
\end{remark}
\begin{remark}
    In practical two-way communication, the parties may need two separately managed keys (one for securing communication from Alice to Bob and the other for securing communication from Bob to Alice).  The two keys could, for example, be obtained through two iterations of the DSKE protocol.  However, for simplicity, we will not discuss this two-way communication case here. After the initial setup phase (i.e., Steps (1) and (2) to be introduced in \Cref{app:protocol}), all subsequent communications in the protocol can be done through insecure classical channels such as the Internet, radio, or phone. Note that no quantum channels are needed in the subsequent communication. This makes DSKE versatile.
\end{remark}

In the DSKE protocol, Alice generates $n$ shares using PSRD shared between Alice and each Security Hub in an $(n, k)$-threshold scheme of Shamir's secret sharing scheme \cite{shamir1979share}, where $k$ is the minimum number of shares needed to reconstruct the secret. She also generates a secret-authenticating tag $o^A:=h'_{u^A}(S^A)$, where $u^A \parallel S^A$ is the secret from the $(n, k)$-threshold scheme, and $h'_{u^A}$ is a hash function with its parameter $u_A$, which is chosen from a family of 2-universal hash functions (also see \cref{eq:hash_family_2}). She encrypts each share and sends the $i$th share $Y_i$ along with the secret-authenticating tag to the Hub $P_i$ via authenticated channels. We note that the secret-authenticating tag is the same for all Hubs. 

After receiving the secret-authenticating tag and the encrypted share, an honest Hub decrypts the share, and then re-encrypts the share using the PSRD shared between the Hub and Bob. It forwards the secret-authenticating tag and the newly encrypted share to Bob via an authenticated channel. 

After Bob receives enough messages from Hubs, he reconstructs all possible values of the secret from any $k$ shares received. Then in the secret validation step, he validates each possible candidate against the secret-authenticating tag, which is chosen to be the same secret-authenticating tag sent by at least $k$ Hubs. If there is no unique secret that passes the secret validation step, he aborts the protocol. 
\subsection{Threat model}\label{sec:threat_model}

A collection of adversarial entities can include a coalition of end users other than Alice and Bob, eavesdroppers, and a subset of the Security Hubs. This set of adversaries may collude to attempt to compromise the objective of the protocol between Alice and Bob. We call this collection of adversarial entities \textit{Eve}. 

An honest Security Hub follows its part of the protocol correctly and maintains confidentiality of its own data. A compromised Security Hub may deviate from the protocol. No limits are placed on compromised Hubs, other than that they do not have access to confidential information nor are they able to modify any information held by honest parties.

To analyze the security (i.e. correctness and confidentiality in \Cref{subsec:ideal_system}) of the DSKE protocol, Eve is allowed to attempt to tamper with the communication by modifying the messages. She is free to listen to all communications except for the initial sharing of tables by honest Security Hubs (conducted through secure channels). Eve can fully control those compromised Security Hubs, including knowing all their tables, and knowing and modifying all messages that come from or are delivered by those compromised Hubs.

As a robustness analysis of a protocol is concerned with an honest implementation of the protocol, which is a modified threat model from that of the correctness and confidentiality analysis, we define the behaviour of Eve in this modified threat model, which we call \textit{passive Eve}: Eve is passive on all communication links, that is, she is allowed to listen to all the communications but she does not tamper; she is still given the ability to fully control compromised Security Hubs. When Eve is passive, we show that the DSKE protocol completes (i.e., does not abort) with a high probability.

\subsection{Assumptions}\label{sec:assumptions}

We list assumptions used in our security proof:
\begin{enumerate}[label = \roman*)]
	\item The pre-shared random data (PSRD) are securely delivered by every Security Hub that is not compromised to each of Alice and Bob. By \textit{securely delivered}, we mean that the confidentiality and integrity of the data is maintained, and that the delivery is assured as being from and to the parties of the correct identity.  This can be achieved through secure channels between Security Hubs and end users. \footnote{PSRD can be delivered by physically shipping a secure data storage device or via QKD links.}
	\item The two users, Alice and Bob, are both honest.
	\item A number of the Security Hubs might be compromised, and this number has a known upper bound. 
	\item For the robustness analysis, a number of the Security Hubs might not correctly execute their part in the protocol, either due to unavailability, communication failure, or compromise, and this number has a known upper bound. This is incorporated as an assumption that Eve is passive on the communication links.  A communication link is assumed to provide the originating identity to the receiver, and if this is incorrectly provided, this is counted as a communication failure.
\end{enumerate}

\FloatBarrier

\section{Security of the skeleton DSKE protocol} \label{sec:security}

In this section, we prove the security of a variant of the DSKE protocol, which we call the \textit{skeleton} DSKE protocol. It differs from the general DSKE protocol (which we summarized in \Cref{subsec:skeleton_protocol_description}; also see \Cref{app:protocol} for a detailed description) only by the assumption that channel authentication is perfectly secure, where a perfectly secure authentication scheme is defined as an $\epsilon$-secure authentication scheme for which $\epsilon = 0$. We call this channel an authenticated channel. In an authenticated channel, an adversary is free to eavesdrop on the communication and to tamper with a message, but the receiver will detect with certainty whether the message has been altered. This may also be thought of as the limiting behaviour as the size of a message authentication code increases without limit.

We analyze the security of the skeleton DSKE protocol in the framework of constructive cryptography as briefly reviewed in \Cref{sec:preliminary}. Further details about the constructive cryptography can be found in \cite{Maurer2011}. Also see \cite{Portmann2022}.

\subsection{Ideal system} \label{subsec:ideal_system}

In the constructive cryptography framework, one can define an ideal system to capture desired functionalities and properties that one hopes to securely realize by the protocol of interest. As the DSKE protocol deals with a multi-party setting where the security claims are against some adversarial subsets, we can require the ideal resource system constructed by the DSKE protocol to respect the security claims that we aim to make for the DSKE protocol. 

In our setting with Alice, Bob, and an adversary Eve, the ideal resource should have three interfaces $A, B$ and $E$. It should produce a secret $S^A$ for Alice and a secret $S^B$ to Bob, which is supposed to be the same as $S^A$ in the case that the protocol completes. Given that Bob can abort in the DSKE protocol when no valid secret can be reconstructed to pass the secret validation step or when there are multiple different reconstructed secrets that pass the secret validation step, we should also allow the ideal resource system to abort under the same conditions which the DSKE protocol does. To indicate that the protocol was aborted, the ideal system sets $S^B$ to be the symbol $\perp$. Security Hubs in the DSKE protocol are treated as resources in our analysis and thus are located inside the ideal resource system. As Eve can control all compromised security Hubs and all communication channels, the ideal system should give Eve the ability to control them through its $E$-interface. Each Security Hub resource has three interfaces, one for the sender, one for the receiver and one for Eve. Each Hub can operate in one of two modes: honest or compromised. In the honest mode, the Hub simply relays the input it receives from the sender to the receiver, and does not output anything at the $E$-interface. In the compromised mode, the Hub outputs the sender's input at the $E$-interface, and uses an input from the $E$-interface to set the output of the receiver's interface. 

While it is only required to specify properties and desired functionalities of the ideal system in the constructive cryptography framework, to better understand the behaviour of this ideal system in the DSKE protocol setting, it is still helpful to think of how this ideal resource system works internally. As shown in \Cref{fig:ideal_resource}, this resource system can be thought as containing a secure key resource to generate a uniformly distributed secret $S^A=S^B$, which will be distributed to Alice and Bob after checking the abort condition (if the protocol aborts, $S^B$ will be set to $\perp$, the flag for aborting the protocol). The ideal resource system internally contains a (modified) real resource system (see \Cref{fig:real_protocol}) that runs the skeleton DSKE protocol with the secret $S^A$ to determine the abort conditions but ignores the reconstructed secret at Bob's side. In particular, it uses the Shamir $(n,k)$-threshold scheme to generate shares by using $S^A$ as the secret. The secret-authenticating tag $o^A$ is computed using the secret $S^A$ as well as an additional key $u^A$ produced by its internal secret key resource. The ideal system gives $S^A$ as its $A$-interface's output and $S^B$ as its $B$-interface's output.

Since our threat model allows Eve to control compromised Hubs, we need to allow the ideal system to give the full access to those compromised Hubs through its $E$-interface. To discuss the inputs and outputs at the $E$-interface of the ideal system, we break the overall $E$-interface into the $E$-interface of each Hub $P_i$. We treat the Hub $P_i$ and its corresponding communication channels as one entity for ease of discussion and we simply say Hub $P_i$ to mean the entire entity. We use $\overline{Y}^E_i$ and $\overline{T}^E_i$ to denote inputs at the Hub $P_i$'s $E$-interface. Similarly, we use $Y^E_i$ and $T^E_i$ to denote the Hub's outputs at the $E$-interface. We use $Z^E_i\in \{00, 01, 10, 11\}$ to denote another input from the Hub $P_i$'s $E$-interface, which can be used to determine the behaviours of two authenticated channels. To fix the meanings, 0 means the authenticated channel transmits the input message faithfully and 1 means it produces an error. For compromised Hubs, $\overline{Y}^E_i$ can take any allowed value for a share and $\overline{T}^E_i$ can take any allowed value for the secret-authenticating tag. A compromised Hub $P_i$ gives its share $Y_i$ and the secret-authenticating tag $o^A$ it received as the outputs to its $E$-interface, that is, $Y^E_i := Y_i$ and $T^E_i := o^A$. (Note that there might be other ancillary information to be passed to $T^E_i$. See the discussion about $T_i$ in \Cref{sec:real_system}.) We further assume $Z^E_i:=00$ for compromised Hubs since this allows Eve to fully control the compromised Hubs by $\overline{Y}^E_i$ and $\overline{T}^E_i$.  For honest Hubs, $\overline{Y}^E_i$ and $\overline{T}^E_i$ can only take the value $\perp$ (which is the only symbol in the allowed alphabet for those variables related to honest Hubs), indicating that Eve cannot control those honest Hubs. Honest Hubs do not leak the information about their shares and thus we set $Y^E_i := \perp$ for those honest Hubs. Note that since the secret-authenticating tag is not encrypted, we can simply set $T^E_i := o^A$ for honest Hubs as well.

Since each Security Hub can be either honest or compromised and our security claims for the DSKE protocol are based on the assumption that the number of compromised Hubs is below some given thresholds, the behaviour of this ideal resource system should respect those assumptions. By our assumptions, we assume there is some fixed set $C \subset \{1, \dots, n\}$ with $|C| < k$, of the identifiers of all compromised Hubs. 

\begin{figure}[t]
	\includegraphics[width=\linewidth]{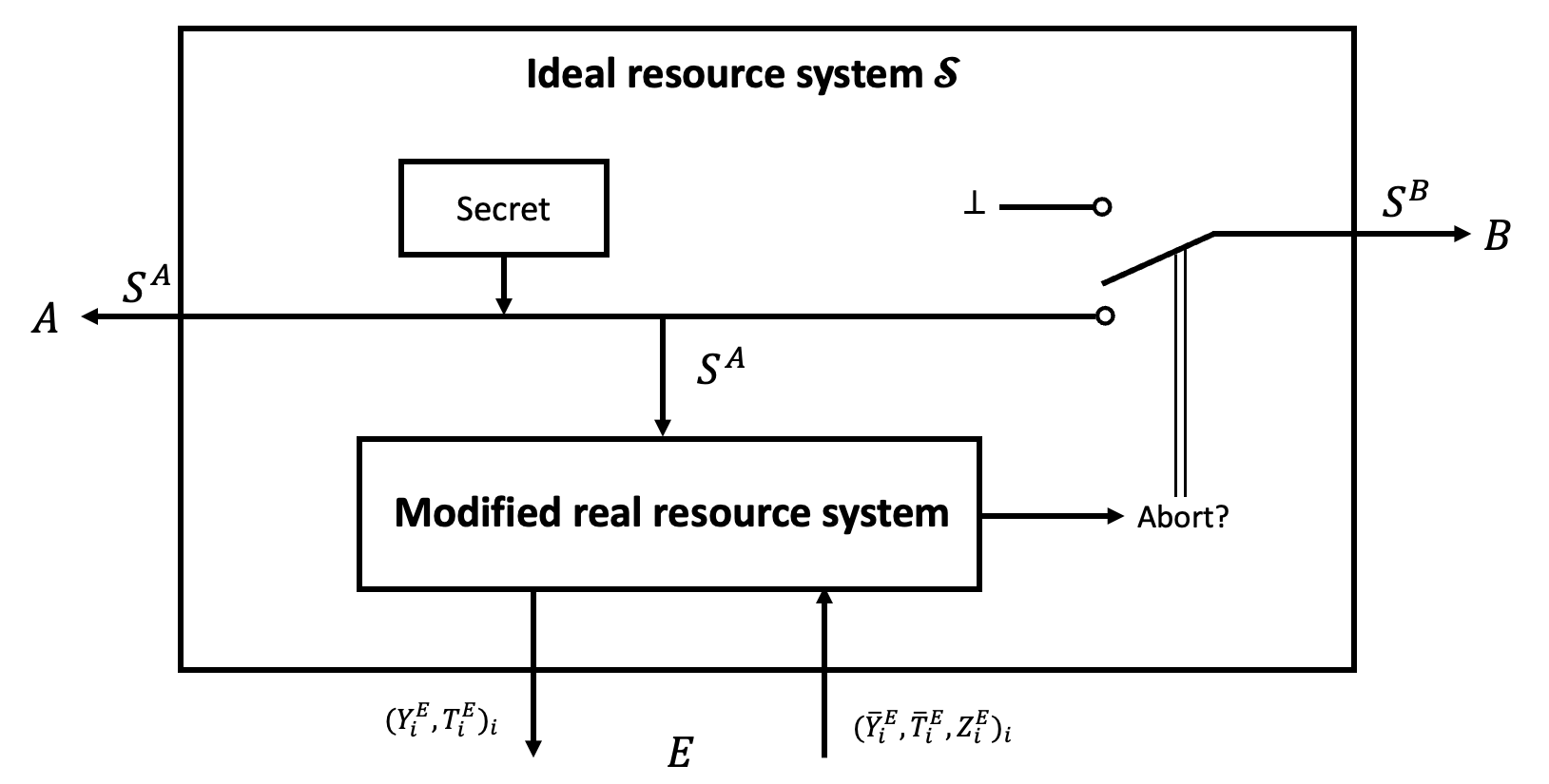}
	\caption{An ideal key distribution resource, which consists of a ``secret'' resource that generates a secret on interfaces $A$ and conditionally $B$ and a modified real resource system that runs the DSKE protocol whose only purpose is to determine whether the protocol aborts. The modified real resource system is the real resource system depicted in \Cref{fig:real_protocol} with the requirement that the secret from running the $(n, k)$-threshold scheme is $S^A$ generated by the secret resource. The ideal system outputs $S^A$ at the $A$-interface, $S^B$ at the $B$-interface. Its $E$-interface is the $E$-interface of the real system with an additional ability (not shown in this diagram) to set the operation mode (i.e., honest versus compromised) of each Security Hub.}
	\label{fig:ideal_resource}
\end{figure}

Here we state some properties of the ideal resource system in terms of the joint probability distribution $Q_{S^A,S^B,(Y^E_i,T^E_i,\overline{Y}^E_i,\overline{T}^E_i, Z^E_i)_i}$ of the inputs and outputs at all its interfaces. This is not an exhaustive list. 
\begin{enumerate}[label=(\roman*).,leftmargin=2em]
	\item Correctness: The marginal distribution of $S^A$ and $S^B$ satisfies that for any $s^A, s^B\in F^m$ such that $s^B \neq s^A$ and $s^B \neq \perp$,%
	\begin{aeq}
		Q_{S^A,S^B}(s^A, s^B) = 0.
	\end{aeq}%
	\item Confidentiality: With $|C| < k$, the conditional probability distribution of $S^A$ conditioned on knowing the values of $(Y^E_j)_{j \in \mc{C}}$ as well as all $T^E_i$ satisfies%
    \begin{aeq}
        Q_{S^A|(Y^E_j)_{j \in \mc{C}}, (T^E_i)_i}=Q_{S^A}. 
    \end{aeq}%
    \item Uniform randomness: The marginal distribution of $S^A$ satisfies that for any $s \in F^m$,%
    \begin{aeq}
        Q_{S^A}(s) = \frac{1}{|F|^m}.
    \end{aeq}%
\end{enumerate}
\begin{remark}
    We note that with the instantiation of the ideal system in \Cref{fig:ideal_resource}, the correctness and uniform randomness of the ideal system are effectively due to the use of a secure key resource (labelled as ``Secret'' in \Cref{fig:ideal_resource}). The confidentiality of the ideal system is effectively due to the confidentiality of the Shamir $(n,k)$-threshold scheme (see \cref{thm:secret_confidentiality}) and that of the secret-authenticating tag (see \cref{thm:hash_function_confidentiality}). This is because its internal modified real system runs the $(n,k)$-threshold scheme using the secret in order to abort under the same condition as the DSKE protocol.
\end{remark}

\subsection{Real system}\label{sec:real_system}

The real system is depicted in \Cref{fig:real_protocol}. In the language of constructive cryptography, the real system uses a set of resources and converters to construct a secure equivalent of the ideal resource system defined in \Cref{subsec:ideal_system}. We define these systems and then discuss how to calculate the distinguishing advantage (a pseudo-metric).

\begin{figure*}[t]
	\includegraphics[width=\linewidth]{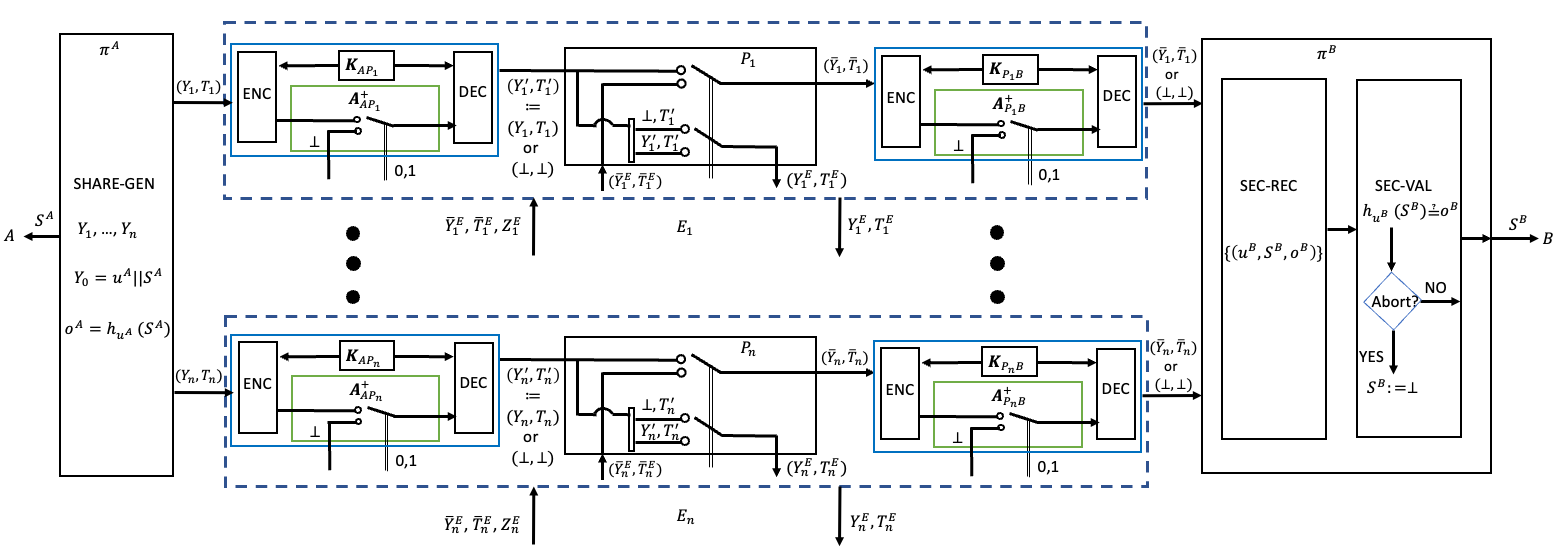}
	\caption{A real key distribution resource using $n$ Security Hubs. Each green box represents an authenticated channel labelled by $\mbf{A}^+$ with a suitable subscript that identifies the communicating parties. The secret key resource $\mbf{K}_{AP_i}$ is used to encrypt $Y_i$ and $\mbf{K}_{P_iB}$ is used to encrypt $\overline{Y}_i$. If a Hub $P_i$ is compromised, Eve determines values of $Y_i, T_i, \overline{Y}_i, \overline{T}_i$ by accepting $(\overline{Y}^E_i, \overline{T}^E_i)$ at the $E$-interface; the system outputs $(Y^E_i, T^E_i):=(Y_i, T_i)$ at the $E$-interface. For honest Hubs, since the ciphertext of $Y_i$ ($\overline{Y}_i$) reveals no information about $Y_i$ ($\overline{Y}_i$), we simply set $Y^E_i:=\perp$ to indicate this. For each Hub $P_i$, a two-bit variable $Z^E_i$ as an input to the $E_i$-interface is used to set the behaviours of two relevant authenticated channels. The $E$-interface of the system consists of $E_1$, \dots, $E_n$, where the alphabets for those inputs and outputs are determined by the operation mode of each corresponding Hub. The operation mode of each Hub is predetermined and cannot be altered through the $E$-interface. Figure legend: ENC: encryption operation, DEC: decryption operation, SHARE-GEN: share generation, SEC-REC: secret reconstruction, SEC-VAL: secret validation.}
	\label{fig:real_protocol}
\end{figure*}

\subsubsection{Resources}

We prove the security of the skeleton protocol assuming the availability of following resources:
\begin{enumerate}[label=(\arabic*)]
    \item Secret key resource: each user and Security Hub pair have a shared secret key resource that has only output interfaces. 
    \item Authenticated channel resource: each communication link between a user and a Security Hub is an authenticated channel, that is, for any message sent through the channel, either the original message is delivered or an error is detected if an adversary attempts to modify the message.  
    \item Security Hub resource: It has three interfaces, one for its sender, one for its receiver and one for Eve. As discussed in \Cref{subsec:ideal_system}, each Security Hub can operate in one of its two modes: honest or compromised. In the honest mode, it receives $(Y_i, T_i)$ from its sender and simply sets $\overline{Y}_i := Y_i$ and $\overline{T}_i := T_i$ to give to its receiver; it ignores inputs from the $E$-interface and sets $T^E_i:=T_i$ and $Y^E_i:=\perp$.  In the compromised mode, it receives $\overline{Y}^E_i$ and $\overline{T}^E_i$ from the $E$-interface, and sets $\overline{Y}_i :=\overline{Y}^E_i$ and $\overline{T}_i :=\overline{T}^E_i$, which are sent to its receiver; it outputs $Y_i$ and $T_i$ received from its sender to the $E$-interface, i.e., $Y^E_i:=Y_i, T^E_i := T_i$. 
\end{enumerate}
We use $\mf{R}_s$ to denote all resources used in the real system.

From a secret key resource and an authenticated channel resource we can construct a secure channel that either transmits the input message confidentially and correctly or produces the $\perp$ symbol to indicate an error when an adversary attempts to modify the message. We remark that in the DSKE protocol, $Y_i$ and $\overline{Y}_i$ are encrypted using the secret key resources. In \Cref{fig:real_protocol} as well as in \Cref{fig:ideal+simulator}, we use $T_i$ to denote the non-encrypted part of the message from Alice to the Hub $P_i$, and similarly $\overline{T}_i$ to denote the non-encrypted part of the message from the Hub $P_i$ to Bob. Then the encrypted version of $Y_i$ (similarly $\overline{Y}_i$) is transmitted together with $T_i$ (correspondingly $\overline{T}_i)$ in one message through an authenticated channel. 

We remark that the secret key resources are available after the one-time setup process in the DSKE protocol. However, the authenticated channel resource is not available in the general DSKE protocol. Thus, after proving the security of the skeleton protocol, we then prove the security of the general DSKE protocol by constructing an authenticated channel resource using a secret key resource and an insecure channel.

\subsubsection{Converters}

We need converters that use secret key resources and authenticated channels and $n$ Security Hub resources to (approximately) construct the ideal system. We now define converters for the DSKE protocol.

As depicted in \Cref{fig:real_protocol}, Alice has a converter $\pi^A$ that produces $S^A$, the shares $Y_i$ and a secret-authenticating tag $o^A$ for validation of the secret. Alice communicates with the Hub $P_i$ through an authenticated channel by sending the encrypted version of $Y_i$ along with $T_i$ where she sets $T_i=P_i \parallel A\parallel B \parallel o^A$, where $P_i \parallel A \parallel B$ is used for identity validation and $o^A$ is the secret-authenticating tag\footnote{Note that $T_i$ and $\overline{T}_i$ in the protocol may contain other necessary information such as the offset for the unused portion of each table and a secret identification number. We omit writing such information here for simplicity as it does not affect our discussion.}. As a result, the Hub can either receive share $Y_i$ and $T_i$ without modification or securely detect errors and get $\perp$ for $Y_i$ and $T_i$. Bob has a converter $\pi^B$ that receives inputs from authenticated channels from each Hub to Bob, and outputs $S^B$, which can take the value of $\perp$ to indicate abort of the protocol. Inside the converter $\pi^B$, it runs the secret reconstruction step and the secret validation step of the DSKE protocol.

The real system is described by a joint probability distribution $P_{S^A,S^B,(Y^E_i,T^E_i,\overline{Y}^E_i,\overline{T}^E_i, Z^E_i)_i}$.

\subsection{Simulator}

As the second condition in \cref{def:composable_security} considers the situation where Eve is active, we need to introduce a simulator $\sigma^E$ (which is Eve's converter) such that when it connects with the ideal system, the $E$-interface of $\mf{S}\sigma^E$ is the same as the $E$-interface of the real system $\pi^s\mf{R}_s$. Note that $\mc{C}$ is used to denote the set of all identifiers of compromised Security Hubs, and that $|C| \leq k-1$ under our assumptions. As Eve can fully control compromised Hubs in our threat model, the simulator $\sigma^E$ needs to have different operations for Hubs in $C$ and for Hubs in the complement set.  

We note that the $E$-interface of the ideal system in \Cref{fig:ideal_resource} accepts $(\overline{Y}^E_i, \overline{T}^E_i, Z^E_i)_i$ as inputs and produces $(Y^E_i, T^E_i)_i$ as outputs. It also gives the ability to set the operation modes of Security Hubs. In the real system shown in \Cref{fig:real_protocol}, Security Hubs whose identifiers are in the set $C$ operate in the compromised mode while the rest are in the honest mode. This means the simulator $\sigma^E$ needs to set the operation modes of those Hubs and blocks the ability to change the operation modes to match the $E$-interface of the real system. Except setting the operation modes, the $E$-interface of the simulator $\sigma^E$ is then identical to the $E$-interface of the ideal system. We depict one possible simulator $\sigma^E$ in \Cref{fig:ideal+simulator} for this purpose. 

\begin{figure*}[t]
	\includegraphics[width=\linewidth]{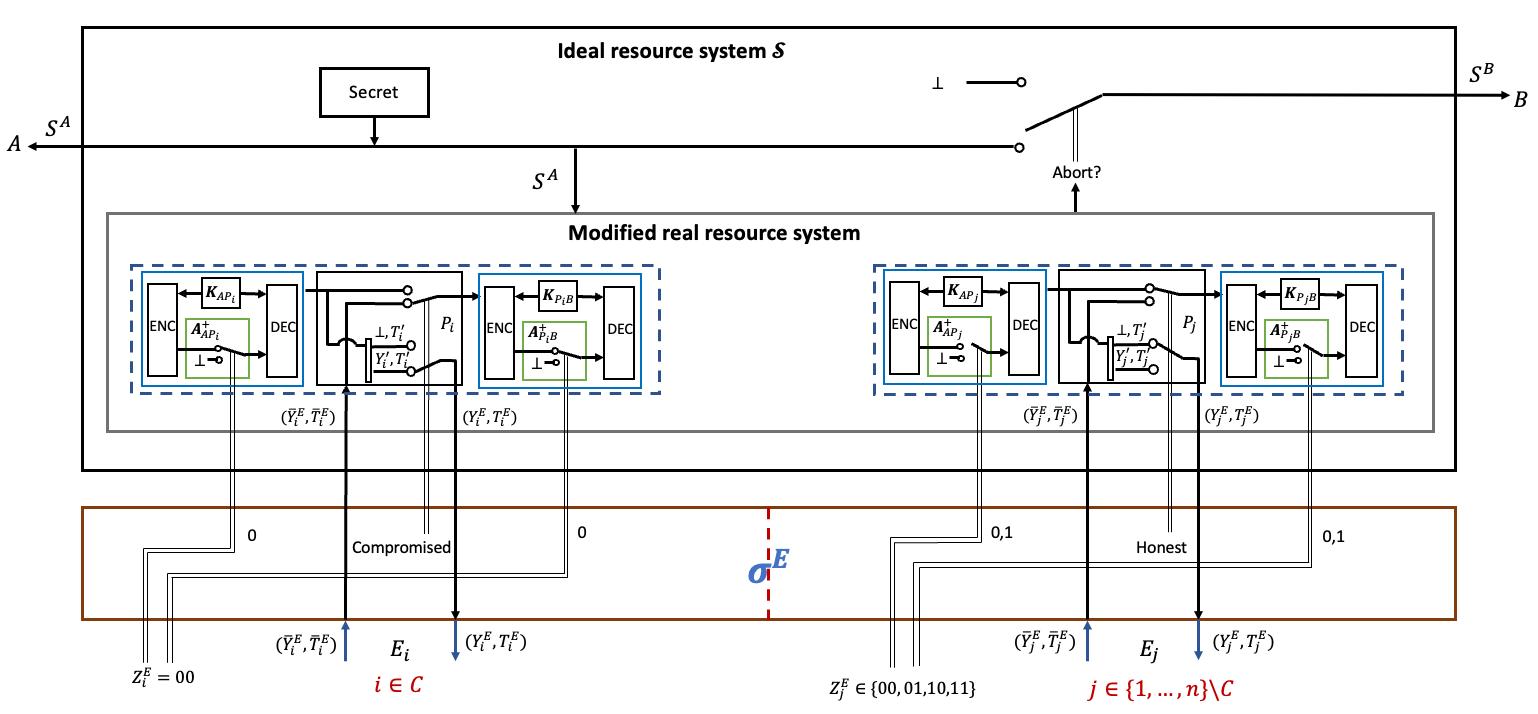}
	\caption{An ideal key agreement resource using $n$ Security Hubs with the simulator $\sigma^E$. Hubs are reordered for drawing purposes only. The set of compromised Hubs is denoted by $C$. The simulator $\sigma^E$ sets operation modes of the Security Hubs contained inside the ideal resource system. For Hubs in the set $C$, it sets the Hub $P_i$ to operate in the compromised mode and uses $Z^E_i:=00$ it receives from its $E$-interface to set the behaviours of two authenticated channels for the Hub $P_i$. For all other Hubs that are not in the set $C$, the simulator $\sigma^E$ sets the Hub $P_j$ to operate in the honest mode and uses $Z^E_j \in \{00, 01, 10, 11\}$ to set the behaviours of two authenticated channels for the Hub $P_j$. The $E$-interface of the simulator also accepts the a pair of values $(\overline{Y}^E_i,\overline{T}^E_i)$ and it outputs $(Y^E_i, T^E_i)$ received from the ideal system. The allowed alphabets for those variables are determined by the operation mode of each corresponding Hub.}
	\label{fig:ideal+simulator}
\end{figure*}

\subsection{Distinguisher}

For $3$-interface resources, a distinguisher $\mf{D}$ is a system with $4$ interfaces, where $3$ interfaces connect to the interfaces of a resource $\mf{R}$ and the other interface outputs a bit, which indicates its guess about which resource is given. This is illustrated in \Cref{fig:distinguisher}. In the DSKE protocol setup, the distinguisher has the following abilities to interact with the unknown system:

\begin{enumerate}[label=(\roman*)]

    \item The distinguisher $\mf{D}$ can read outputs from the $A$- and $B$-interfaces. These two interfaces do not accept any input. The distinguisher can set the inputs of and read outputs from the $E$-interface unless the $E$-interface has no input/output when a simulator is used to block input/output.

    \item When the $E$-interface allows the control of internal communication links, the distinguisher $\mf{D}$ can attack internal communication as allowed by Eve's ability since the distinguisher can act as Eve through the $E$-interface.  

    \item When the $E$-interface allows the control of compromised Security Hubs, the distinguisher $\mf{D}$ can control those compromised Security Hubs.

\end{enumerate}

As the real system is completely characterized by $P_{S^A,S^B,(Y^E_i,T^E_i,\overline{Y}^E_i,\overline{T}^E_i, Z^E_i)_i}$ and the ideal system with the simulator $\sigma^E$ is completely characterized by $Q_{S^A,S^B,(Y^E_i,T^E_i,\overline{Y}^E_i,\overline{T}^E_i, Z^E_i)_i}$, the distinguishing advantage (pseudo-metric) defined in \cref{eq:distinguishing_advantage} can be related to the statistical distance between $P_{S^A,S^B,(Y^E_i,T^E_i,\overline{Y}^E_i,\overline{T}^E_i, Z^E_i)_i}$ and $Q_{S^A,S^B,(Y^E_i,T^E_i,\overline{Y}^E_i,\overline{T}^E_i, Z^E_i)_i}$ as we will see in the security proof.

\begin{figure}[ht]
    \includegraphics[width=\linewidth]{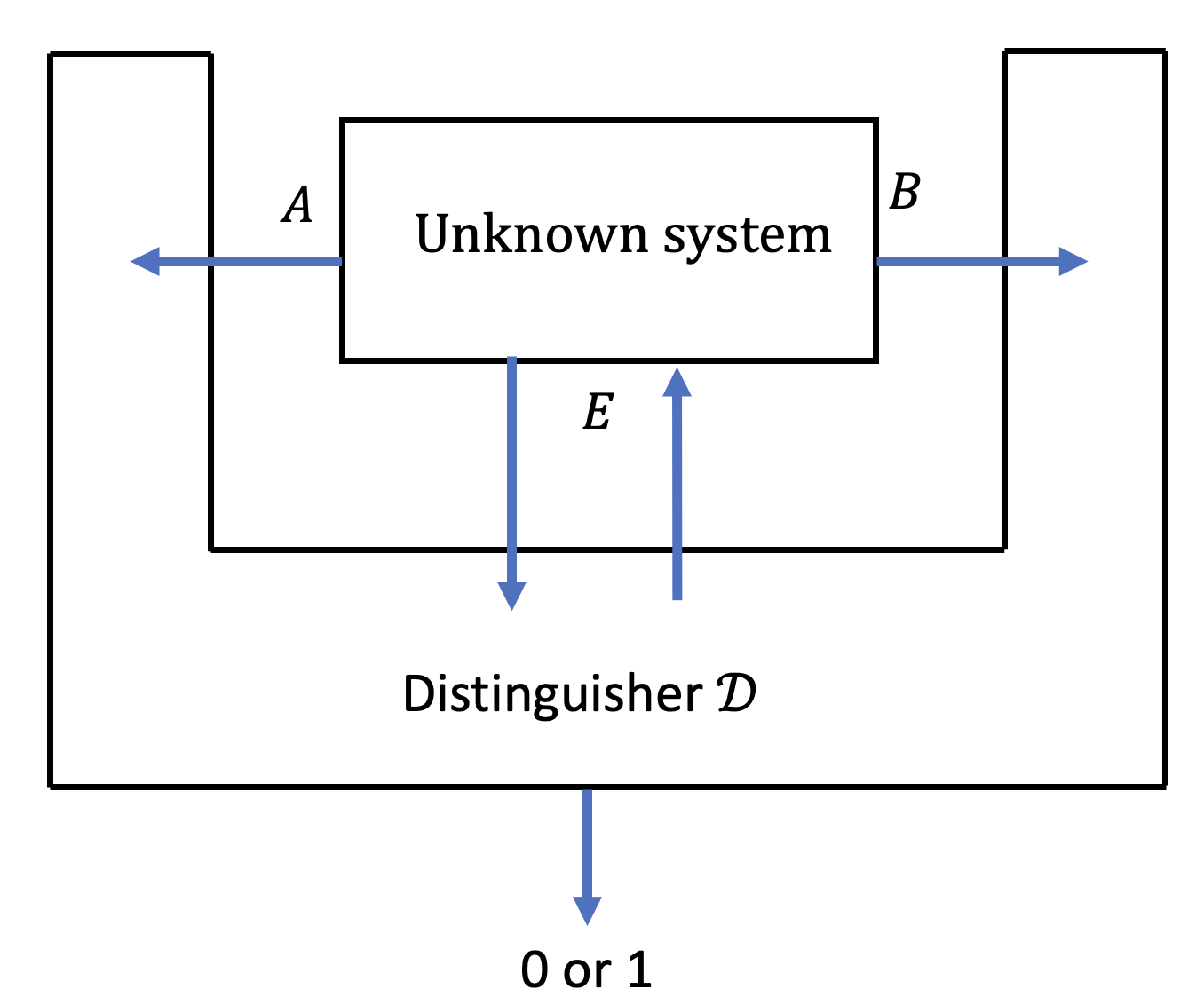}
    \caption{\label{fig:distinguisher}The distinguisher $\mf{D}$ connects to an unknown system, interacts with it and reproduces a one-bit output that indicates its guess about the identity of the unknown system.}
\end{figure}

\subsection{Security}

\begin{theorem}\label{thm:skeleton_security}
	Under the assumptions listed in \Cref{sec:assumptions}, the protocol $\pi^s = (\pi^A, \pi^B)$ described above (depicted in \Cref{fig:real_protocol}) is $\epsilon$-secure, where $\epsilon=\min({n \choose k}\frac{m+1}{|F|},1)$, which is determined by the family of hash functions used in the protocol as described in \cref{thm:secret_validation} as well as the choices of $n$ and $k$ in the $(n, k)$-threshold scheme. 
\end{theorem}
\begin{proof}
	We need to check two conditions in \cref{def:composable_security}.

	To check the first condition, we use converters $\alpha^E$ and $\gamma^E$ to plug into the $E$-interface in the real and ideal systems. They both allow all messages $Y_i, T_i$ and $\overline{Y}_i$, $\overline{T}_i$ to be delivered correctly in authenticated channels (that is, $Z^E_i = 00$ for all $i$) and completely block $E$-interface from the distinguisher. Those two converters may still control compromised Hubs in the same predefined way with the restriction that the number of compromised Hubs is at most $k-1$. In this case, the distinguisher can observe only outputs from $A$- and $B$-interfaces.  

    We use $X$ to denote $S^A,S^B$ and use $\mathcal{X}$ to denote the set of values that $S^A,S^B$ can take. As any distinguisher can only observe $S^A$ and $S^B$, from the distinguisher's point of view, the real system is completely described by $P_{X}$ and the ideal system with the simulator $\sigma^E$ by $Q_{X}$. To guess whether it is holding the real system, a deterministic strategy for the distinguisher is that the distinguisher can pick a subset $\mathcal{X}' \subseteq \mathcal{X}$ such that for all $x \in \mathcal{X}'$, it outputs $1$ and for all other values of $x$, it outputs 0. The distinguisher may also choose a mixed (probabilistic) strategy. We note that each mixed strategy is just a probabilistic mixture of pure (deterministic) strategies. Let $\mf{D}'$ be a distinguisher that uses an arbitrary mixed strategy which is a probabilistic mixture of deterministic strategies $\mathcal{X}'_k$ with corresponding probabilities $p_k$. The distinguishing advantage of this strategy is thus bounded by%
    \begin{aeq}
        &\Pr[\mf{D}'\pi^s\mf{R}_s\alpha^E=1]-\Pr[\mf{D}'\mf{S}\gamma^E=1] \\
        =~&\sum_k p_k \sum_{x \in \mathcal{X}'_k} (P_X(x) - Q_X(x))\\
        \leq~& \sum_k p_k \Big[\max_{\mathcal{X}':\mathcal{X}'\subseteq \mathcal{X}}\sum_{x \in \mathcal{X}'} (P_X(x) - Q_X(x))\Big]\\
        =~&\max_{\mathcal{X}':\mathcal{X}'\subseteq \mathcal{X}}\sum_{x \in \mathcal{X}'} (P_X(x) - Q_X(x)),
    \end{aeq}%
    where the first equality is due to the chosen strategy of the distinguisher $\mf{D}'$, which can be written as a probabilistic mixture of deterministic strategies, the inequality is due to the fact that we perform an optimization over all subsets $\mathcal{X}'$ of $\mathcal{X}$ and $\mathcal{X}'_k$ is just a possible subset of $\mathcal{X}$, and the last equality is the result of summing over $k$. Thus, it is enough to consider all deterministic strategies. 

    In this case, the distinguishing advantage is%
 	\begin{aeq}\label{eq:ske_metric_calculation}
		&d(\pi^s\mf{R}_s\alpha^E,\mf{S}\gamma^E)  \\
        =~& \max_{\mf{D} \in \mathbb{D}}(\Pr[\mf{D}\pi^s\mf{R}_s\alpha^E=1]-\Pr[\mf{D}\mf{S}\gamma^E=1]) \\
        =~& \max_{\mathcal{X}':\mathcal{X}'\subseteq \mathcal{X}}\sum_{x \in \mathcal{X}'} (P_X(x) - Q_X(x)) \\
        =~& \frac{1}{2}\sum_{x \in \mathcal{X}}  |P_{X}(x) - Q_{X}(x)|, \\
    \end{aeq}%
    where the first equality is due to the definition of the pseudo-metric in \cref{eq:distinguishing_advantage}, the second equality is due to the general deterministic strategy of the distinguisher as described above and that the optimal value of the distinguishing advantage can always be realized by a deterministic strategy, and the third equality is due to \cref{eq:statistical_distance_property}. We see that the distinguishing advantage is related to the statistical distance between $P_X$ and $Q_X$. 

    Our task is then to evaluate the statistical distance. We note that the allowed alphabets for $\overline{Y}^E_i$ and $\overline{T}^E_i$ depend on whether $i \in \mc{C}$. For ease of writing, we do not explicitly write out alphabets for those variables.  
    For any value of $s^A$ and any value of $s^B$ such that $s^B \neq s^A$ and $s^B \neq \perp$, the joint probability distribution of the ideal system $Q_{X}(s^A, s^B) = 0$, while the joint probability distribution of the real system $P_{X}(s^A, s^B) \neq 0$ due to the correctness of secret validation being approximate, which depends on the property of hash functions used for the secret-authenticating tag as in \cref{thm:secret_validation}. When restricting to all possible values of $(s^A, s^B)$ such that $P_X(s^A, s^B)> Q_X(s^A, s^B)$, the real and ideal systems differ only in the case where the real system may obtain $s^B \neq s^A$ and $s^B \neq \perp$. To see this, we show the contrapositive: if $s^A=s^B$ or $s^B =\perp$, then $P_{X}(s^A, s^B) \leq Q_{X}(s^A, s^B)$. For each $s^A$, we observe (i) $P_{S^A}(s^A) = Q_{S^A}(s^A) = \frac{1}{|F|^m}$, (ii) $P_{X}(s^A, \perp) = Q_{X}(s^A, \perp)$ and (iii) for each $s^B \neq s^A$, $Q_{X}(s^A, s^B) = 0$ while $P_{X}(s^A, s^B) \geq 0$. As we can write
    \begin{aeq}
P_{S^A}(s^A) &= \sum_{s^B} P_X(s^A, s^B)\\ &=  \sum_{s^B \neq \perp} P_X(s^A, s^B)  +  P_X(s^A, \perp),   
    \end{aeq}% 
    and
        \begin{aeq}
Q_{S^A}(s^A) &= \sum_{s^B} Q_X(s^A, s^B) \\&=  Q_X(s^A, s^A)  +  Q_X(s^A, \perp),
    \end{aeq}%
    these two equations imply $P_X(s^A, s^A) \leq Q_X(s^A, s^A)$ for each $s^A$ following those facts (i)-(iii).

    Thus, the distinguishing advantage is% 

    \begin{flalign}
		&d(\pi^s\mf{R}_s\alpha^E,\mf{S}\gamma^E) \nonumber &\\
		=~& \sum_{x: P_X(x) \geq Q_X(x)} (P_X(x) - Q_X(x)) \label{eq:cond1_ske_calc_distance_property} &\\
		=~& \sum_{s^A} \sum_{s^B: s^B \neq s^A, s^B \neq \perp} P_{X}(s^A, s^B) \label{eq:cond1_ske_calc_P_Q_properties} &\\
		=~& \sum_{s^A, (\overline{y}_i, \overline{t}_i)_i}  \sum_{\substack{s^B: s^B \neq s^A, \\ s^B \neq \perp}}P_{S^A, S^B, (\overline{Y}^E_i, \overline{T}^E_i)_i}(s^A, s^B, (\overline{y}_i, \overline{t}_i)_i) \label{eq:cond1_ske_cal_sum_over_yt} &\\
		=~& \sum_{s_A, (\overline{y}_i, \overline{t}_i)_i} P_{S^A, (\overline{Y}^E_i, \overline{T}^E_i)_i}(s^A, (\overline{y}_i, \overline{t}_i)_i) \nonumber &\\
		& \times   \sum_{ \substack{s^B: s^B \neq s^A,\\ s^B \neq \perp}} P_{S^B|S^A, (\overline{Y}^E_i, \overline{T}^E_i)_i}(s^B | s^A, (\overline{y}_i, \overline{t}_i)_i) \label{eq:cond1_ske_cal_convert_to_conditional_prob}, &
	\end{flalign}%
    where we use \cref{eq:statistical_distance_property} to obtain \cref{eq:cond1_ske_calc_distance_property} from \cref{eq:ske_metric_calculation}, we obtain \cref{eq:cond1_ske_calc_P_Q_properties} since the condition $P(x)\geq Q(x)$ is equivalent to the situation where Bob does not abort the protocol and Bob obtains $S^B$ that is different from $S^A$ as discussed above (also note $Q_X(x) = 0$ for such $x$), we write the marginal probability $P_{S^A, S^B}$ in terms of summing $P_{S^A, S^B, (\overline{Y}^E_i, \overline{T}^E_i)_i}$ over all possible values of $(\overline{Y}^E_i, \overline{T}^E_i)_i$ to get \cref{eq:cond1_ske_cal_sum_over_yt}, and we finally rewrite the joint probability over $S^A, S^B, (\overline{Y}^E_i, \overline{T}^E_i)_i$ by the conditional probability $P_{S^B | S^A, (\overline{Y}^E_i, \overline{T}^E_i)_i}$ and the marginal probability $P_{S^A, (\overline{Y}^E_i, \overline{T}^E_i)_i}$ to obtain \cref{eq:cond1_ske_cal_convert_to_conditional_prob}.

    Under the assumption that at most $k-1$ Hubs are compromised, we can safely assume whenever $S^B \neq \perp$, the secret-authenticating tag $o^A$ is transmitted faithfully from at least $k$ Hubs to Bob. The reason is that Bob would set $S^B$ to $\perp$ and abort the protocol if fewer than $k$ Hubs send the same secret-authenticating tag; on the other hand, to agree on a value $t \neq \perp$ other than $o^A$ for the secret-authenticating tag, at least $k$ Hubs need to send the same modified value $t$, which is not possible given that at most $k-1$ Hubs are compromised. Thus, for each value of $s^A, (\overline{y}_i, \overline{t}_i)_i$,%
    \begin{aeq}\label{eq:con_prob_epsilon_estimation}
        & \sum_{ \substack{s^B: s^B \neq s^A, \\ 
        s^B \neq \perp}} P_{S^B|S^A, (\overline{Y}^E_i, \overline{T}^E_i)_i}(s^B | s^A, (\overline{y}_i, \overline{t}_i)_i) \\
    	\leq & \min({n \choose k}\Pr[h_{u'}(s^B) = o^A|s^A, (\overline{y}_i, \overline{t}_i)_i, s^B \neq s^A], 1) \\
    	\leq &  \min({n \choose k}\frac{m+1}{|F|},1)=:\epsilon,
    \end{aeq}%
    where $u' \parallel s^B$ is the secret reconstructed from a subset of $\{\overline{y}_i\}$ with $k$ elements, the factor ${n \choose k}$ is due to there being up to {${n \choose k}$} possible values of $u'\parallel s^B$ that the secret validation step can check, and we use \cref{thm:secret_validation} (where we remove the freedom to choose $t'$ and set $t'=0$) for the last inequality.%
    \footnote{The multiplier ${n \choose k}$ may significantly reduce the security for large $n$.  For example, $n=99$, $k=50$ results is a security loss of $\log {99\choose 50}=95.35\text{ bits}$.  For $n\le11$, the security loss is under 9 bits, but with Bob still having to search up to ${11 \choose 6}=462$ combinations.  A protocol variant, which we do not elaborate, allows Bob to filter out bad shares before combining them, at the cost of several tags per share.}

    Combining \cref{eq:cond1_ske_cal_convert_to_conditional_prob,eq:con_prob_epsilon_estimation}, we have%
    \begin{aeq}
    	&d(\pi^s\mf{R}_s\alpha^E,\mf{S}\gamma^E) \\
    	\leq & \sum_{s_A, (\overline{y}_i, \overline{t}_i)_i} P_{S^A, (\overline{Y}^E_i, \overline{T}^E_i)_i}(s^A, (\overline{y}_i, \overline{t}_i)_i) \epsilon
    	= \epsilon.
    \end{aeq}%
    This verifies the first condition in \cref{def:composable_security}.

    To check the second condition, we consider the simulator $\sigma^E$ as depicted in \Cref{fig:ideal+simulator}. In this case, Eve's interface allows the distinguisher to control all compromised Hubs. A distinguisher can pick any allowed value of $(\overline{Y}^E_i, \overline{T}^E_i, Z^E_i)_i$ as an input to the system and observe $S^A, S^B, (Y^E_i, T^E_i)_i$ from outputs of the unknown system. The restriction on the number of compromised Hubs is reflected by different alphabets of $\overline{Y}^E_i$, $\overline{T}^E_i$ and $Z^E_i$ for $i \in \mc{C}$ and for $i \in \{1, \dots, n\}\setminus\mc{C}$. The real system is completely characterized by $P_{S^A,S^B,(Y^E_i,T^E_i,\overline{Y}^E_i,\overline{T}^E_i, Z^E_i)_i}$ and the ideal system with the simulator $\sigma^E$ is completely characterized by $Q_{S^A,S^B,(Y^E_i,T^E_i,\overline{Y}^E_i,\overline{T}^E_i, Z^E_i)_i}$. We now use $X$ to denote $S^A,S^B,(Y^E_i,T^E_i,\overline{Y}^E_i,\overline{T}^E_i, Z^E_i)_i$ and let $\mathcal{X}$ denote the set of values that $S^A,S^B,(Y^E_i,T^E_i,\overline{Y}^E_i,\overline{T}^E_i, Z^E_i)_i$ can take.

    We note that the ideal system and the real system abort under the same condition since the ideal system internally runs the DSKE protocol to determine abort conditions. Therefore, for any allowed value $s^A, (y_i, t_i, \overline{y}_i, \overline{t}_i, z_i)_i$ for $(Y^E_i,T^E_i,\overline{Y}^E_i,\overline{T}^E_i, Z^E_i)_i$,%
    \begin{aeq}
    	P_{X}(s^A, \perp, (y_i, t_i, \overline{y}_i, \overline{t}_i, z_i)_i) = Q_{X}(s^A, \perp, (y_i, t_i, \overline{y}_i, \overline{t}_i, z_i)_i).
    \end{aeq}%

    By a similar argument as for the first condition, it is enough to consider all deterministic strategies when we calculate the pseudo-metric for these two systems. To guess that it is holding the real system, the distinguisher can pick a subset $\mathcal{X}' \subseteq \mathcal{X}$ such that for all $x \in \mathcal{X}'$, it outputs $1$ and for all other values of $x$, it outputs 0. (Note that we reuse the same notation as in the proof of the first condition since the proof of the second condition resembles the first one. However, $X, \mathcal{X}$ here denote a different combination of symbols and a different set, respectively.) The distinguishing advantage is
    \begin{flalign}
		&  d(\pi^s\mf{R}_s,\mf{S}\sigma^E) \nonumber &\\
		= & \sum_{x: P_X(x) \geq Q_X(x)} (P_X(x) - Q_X(x)) \label{eq:ske_calc_distance_property} &\\
		= & \sum_{\substack{s^A, \\ (y_i, t_i, \overline{y}_i, \overline{t}_i, z_i)_i}} \sum_{ \substack{s^B: s^B \neq s^A,\\ s^B \neq \perp}} P_{X}(s^A, s^B, (y_i, t_i, \overline{y}_i, \overline{t}_i, z_i)_i) \label{eq:ske_calc_P_Q_properties} &\\
		= & \sum_{s^A, (\overline{y}_i, \overline{t}_i, z_i)_i}  \sum_{\substack{s^B: s^B \neq s^A, \\ s^B \neq \perp}}P_{S^A, S^B, (\overline{Y}^E_i, \overline{T}^E_i, Z^E_i)_i}(s^A, s^B, (\overline{y}_i, \overline{t}_i, z_i)_i) \label{eq:ske_cal_sum_over_yt}&\\
		=& \sum_{s_A, (\overline{y}_i, \overline{t}_i,z_i)_i} P_{S^A, (\overline{Y}^E_i, \overline{T}^E_i, Z^E_i)_i}(s^A, (\overline{y}_i, \overline{t}_i)_i, z_i) \nonumber &\\
		& \times   \sum_{ \substack{s^B: s^B \neq s^A,\\ s^B \neq \perp}} P_{S^B|S^A, (\overline{Y}^E_i, \overline{T}^E_i, Z^E_i)_i}(s^B | s^A, (\overline{y}_i, \overline{t}_i, z_i)_i) \label{eq:ske_cal_convert_to_conditional_prob}, &
	\end{flalign}%
    where we obtain \cref{eq:ske_calc_P_Q_properties} since the condition $P(x)\geq Q(x)$ is equivalent to the situation where Bob does not abort the protocol and Bob obtains $S^B$ that is different from $S^A$ as discussed above (also note $Q_X(x) = 0$ for those $x$'s), to get \cref{eq:ske_cal_sum_over_yt}, we directly sum over all possible values of $(Y^E_i, T^E_i)_i$ since those are outputs from $E$-interface that are not used for the secret validation step, and we finally rewrite the joint probability over $S^A, S^B, (\overline{Y}^E_i, \overline{T}^E_i, Z^E_i)_i$ by the conditional probability $P_{S^B | S^A, (\overline{Y}^E_i, \overline{T}^E_i, Z^E_i)_i}$ and the marginal probability $P_{S^A, (\overline{Y}^E_i, \overline{T}^E_i, Z^E_i)_i}$ to obtain \cref{eq:ske_cal_convert_to_conditional_prob}. 

    We note that \cref{eq:con_prob_epsilon_estimation} also holds when we condition on the additional input choices $(Z^E_i)_i$ that determine the behaviours of authenticated channels since \cref{eq:con_prob_epsilon_estimation} was proved under the choice $Z^E_i =00$ for all $i$ and we considered all possible combinations of $u' \parallel s^B$ that can go through to the secret validation step when we estimated that upper bound. By setting any of $(Z^E_i)_i$ to any value other than $00$, Eve effectively reduces the number of combinations of $u' \parallel s^B$ to feed into the secret validation step while the probability for each combination of $u' \parallel s^B$ to pass the secret validation step is unchanged. In other words, for each allowed value $s^A, (\overline{y}_i, \overline{t}_i, z_i)_i$ of $S^A, (\overline{Y}^E_i, \overline{T}^E_i, Z^E_i)_i$,%
    \begin{aeq}\label{eq:cond2_cond_prob_estimation}
        & \sum_{ \substack{s^B: s^B \neq s^A, \\ 
        s^B \neq \perp}}P_{S^B|S^A, (\overline{Y}^E_i, \overline{T}^E_i, Z^E_i)_i}(s^B | s^A, (\overline{y}_i, \overline{t}_i, z_i)_i) \leq \epsilon.
    \end{aeq}%

    By using \cref{eq:cond2_cond_prob_estimation}, the distinguishing advantage is then%
    \begin{aeq}\label{eq:skeleton_secrecy_metric_calculation}
    	&d(\pi^s\mf{R}_s,\mf{S}\sigma^E) \\
    	=& \sum_{s_A, (\overline{y}_i, \overline{t}_i, z_i)_i} P_{S^A, (\overline{Y}^E_i, \overline{T}^E_i, Z^E_i)_i}(s^A, (\overline{y}_i, \overline{t}_i, z_i)_i) \\
    	& \times   \sum_{ \substack{s^B: s^B \neq s^A,\\ s^B \neq \perp}} P_{S^B|S^A, (\overline{Y}^E_i, \overline{T}^E_i, Z^E_i)_i}(s^B | s^A, (\overline{y}_i, \overline{t}_i, z_i)_i) \\
    	\leq & \sum_{s_A, (\overline{y}_i, \overline{t}_i, z_i)_i} P_{S^A, (\overline{Y}^E_i, \overline{T}^E_i, Z^E_i)_i}(s^A, (\overline{y}_i, \overline{t}_i, z_i)_i) \epsilon
        = \epsilon.
    \end{aeq}%
	
	Therefore, the second condition is also verified. Thus, the protocol is $\epsilon$-secure.
\end{proof}
\begin{remark}
    The secret validation step relies on the correctness of the secret-authenticating tag as stated in \cref{thm:secret_validation}. In the limit that $|F|\to\infty$, we have that $\epsilon\to0$, and the secret validation step is perfectly correct. In this case, the skeleton DSKE protocol perfectly constructs the ideal system out of the real resource system. 
\end{remark}

\subsection{Robustness}\label{sec:skeleton_robustness}

Robustness of a protocol is the condition that the protocol does not abort when Eve is passive (restricted to modifying the messages of compromised Hubs only), as defined in \Cref{sec:threat_model}.  We say a protocol is $\epsilon$-robust if its aborting probability is at most $\epsilon$ in the case of this restriction. 
When restricted in this way, the aborting probability for the skeleton DSKE protocol is at most $\epsilon=\min({n \choose k}\frac{m+1}{|F|},1)$ when the number of honest Security Hubs is at least $k$ and the number of compromised Security Hubs is at most $k-1$.\footnote{The second condition becomes significant when $2k\le n$, due to multiple competing reconstructed secrets.}
\begin{theorem}[Robustness of a skeleton DSKE protocol]\label{thm:robustness_skeleton}
	Under the assumptions listed in \Cref{sec:assumptions}, when the upper bound on the number of compromised Security Hubs is no greater than $\min(n-k, k-1)$, the skeleton DSKE protocol is $\epsilon$-robust with $\epsilon = \min({n \choose k} \frac{m+1}{|F|}, 1)$.
\end{theorem}
\begin{proof}
	When Eve is passive, all authenticated channels in the skeleton DSKE protocol faithfully transmit messages. We note that a message using a fake identity gets rejected by the receiver without consuming any secret key resources by the assumption that a communication link provides the originating identity to the receiver. In other words, Eve cannot impersonate an honest Security Hub by using compromised Hubs to mount an attack to exhaust data in a table shared by a Security Hub and a user.
 
    If the number of honest Security Hubs is at least $k$ (which means the number of compromised Hubs cannot be more than $n-k$) and the number of compromised Hubs is at most $k-1$, the only condition that causes Bob to abort is when there are multiple different candidate tuples $(u^A, s^A, o^A)$ that pass the secret validation step. That is, in addition to the correct secret reconstructed using shares from $k$ honest Security Hubs, there must exist at least one different secret candidate that also passes the secret validation step. For each possible guess of the $(u^A, s^A)$ other than the correct secret, it can pass the secret validation step with a probability at most $\min(\frac{m+1}{|F|},1)$ due to \cref{thm:secret_validation}. Since there are at most ${n \choose k}$ valid candidates to go through the secret validation step, the aborting probability in this case is at most $\min({n \choose k}\frac{m+1}{|F|},1)=: \epsilon$. Thus, the skeleton DSKE protocol is $\epsilon$-robust when the upper bound on the number of compromised Security Hubs is no greater than $\min(n-k, k-1)$.
\end{proof}

\FloatBarrier

\section{Security of the general DSKE protocol}\label{sec:proof}

To assist our discussion here, as shown in \Cref{fig:real_protocol}, we denote the secure key resource between Alice and the Hub $P_i$ as $\mbf{K}_{AP_i}$, the authenticated channel between Alice and the Hub $P_i$ as $\mbf{A}^{+}_{AP_i}$, the secure key resource between the Hub $P_i$ as $\mbf{K}_{P_iB}$, and the authenticated channel between the Hub $P_i$ and Bob as $\mathbf{A}^{+}_{P_iB}$. In the previous section, we have shown that $\pi^s$ securely constructs the ideal resource $\mf{S}$ out of $\mf{R}_s$ within $\epsilon$, where the real resource $\mf{R}_s$ in the skeleton DSKE protocol is defined as%
\begin{aeq}
	\mf{R}_s := & \mbf{K}_{AP_1} \parallel \mbf{A}^{+}_{AP_1} \parallel P_1 \parallel \mbf{K}_{P_1B} \parallel \mbf{A}^{+}_{P_1B}\parallel \dots \\
    & \parallel \mbf{K}_{AP_n} \parallel \mbf{A}^{+}_{AP_n} \parallel P_n \parallel\mbf{K}_{P_nB} \parallel \mbf{A}^{+}_{P_nB}.  
\end{aeq}%

The difference between a skeleton protocol and a general protocol using the same $(n, k)$-threshold scheme is the availability of authenticated channel resources. In a general DSKE protocol, we need to use a secure key resource and an insecure channel to construct an authenticated channel. An authentication protocol $\pi^{\auth}_{AP_i}$ (also $\pi^{\auth}_{P_iB}$) using a Carter--Wegman universal hash function family to produce message tags can securely construct an authenticated channel out of a secret key resource and an insecure channel within $\epsilon':=\min(\frac{s}{|F|},1)$ as shown in \Cref{thm:hash_function_correctness}. To distinguish this secret key resource from the secret key resource used in the skeleton protocol, we denote the new secret key resource by $\mbf{K}^{\auth}$ and the insecure channel by $\mbf{C}$ (with suitable subscripts). We then define the resource $\mf{R}_g$ used in the general DSKE protocol as%
\begin{aeq}
	\mf{R}_g :=& \; \mbf{K}_{AP_1} \parallel \mbf{K}^{\auth}_{AP_1} \parallel \mbf{C}_{AP_1} \parallel P_1 \\ & \parallel
    \mbf{K}_{P_1B} \parallel \mbf{K}^{\auth}_{P_1B} \parallel \mbf{C}_{P_1B}  \parallel \dots \\ & \parallel \mbf{K}_{AP_n} \parallel \mbf{K}^{\auth}_{AP_n} \parallel \mbf{C}_{AP_n} \parallel P_n \\ &\parallel \mbf{K}_{P_nB} \parallel \mbf{K}^{\auth}_{P_nB} \parallel \mbf{C}_{P_nB} . 
\end{aeq}%

Recall that $m$ is the length of the secret to be agreed and $s$ is the length of the message for which a tag is produced. We now show that the general DSKE protocol is $(\epsilon + 2n \epsilon')$-secure, where $\epsilon = \min({n \choose k}\frac{m+1}{|F|}, 1)$ and $\epsilon' = \min(\frac{s}{|F|},1)$.

\begin{theorem}
	The protocol $\pi^g = (\pi^s, \pi^{\auth}_{AP_1}, \dots, \pi^{\auth}_{AP_n}, \pi^{\auth}_{P_1B}, \dots \pi^{\auth}_{P_nB})$ securely constructs the ideal resource $\mf{S}$ out of the resource $\mf{R}_g$ within $\epsilon + 2n \epsilon'$, where each of $\pi^{\auth}_{AP_i}$ and $\pi^{\auth}_{P_iB}$ securely construct authenticated channels $\mbf{A}^+_{AP_i}$, $\mbf{A}^+_{P_iB}$ within $\epsilon'$, respectively, where $\epsilon = \min({n \choose k}\frac{m+1}{|F|}, 1)$ and $\epsilon' = \min(\frac{s}{|F|},1)$. 
\end{theorem}
\begin{proof}
	We start with showing the first condition in \cref{def:composable_security} where no adversary is present. We trivially have%
	\begin{aeq}
		\pi^{\auth}(\mbf{K}^{\auth}\parallel\mbf{C}) \alpha^E = \mbf{A}^{+} \beta^E
	\end{aeq}%
    for each choice of subscript ($AP_i$ or $P_iB$), where $\alpha^E$ and $\beta^E$ emulate an honest behaviour at Eve's interfaces, since both systems are equivalent to a channel that faithfully transmits a message between two parties (either $A$ and $P_i$ or $P_i$ and $B$). With a slight abuse of notation, we use $\alpha^{E^{2n}}$ and $\beta^{E^{2n}}$ to denote systems that emulate an honest behaviour at Eve's interfaces in $\pi^g\mf{R}_g$ and $\pi^s\mf{R}_s$, which include $2n$ authentication channels in the systems $\mf{R}_g$ and $\mf{R}_s$, respectively. We use $\gamma^{E^{2n}}$ to denote the system that emulates an honest behaviour at Eve's interfaces in $\mf{S}$. Thus,%
	\begin{aeq}
		& d(\pi^g \mf{R}_g \alpha^{E^{2n}}, \mf{S}\gamma^{E^{2n}}) \\
		\leq & d(\pi^g \mf{R}_g\alpha^{E^{2n}}, \pi^s \mf{R}_s \beta^{E^{2n}}) + d( \pi^s \mf{R}_s \beta^{E^{2n}}, \mf{S}\gamma^{E^{2n}})\\
		= &~ d( \pi^s \mf{R}_s \beta^{E^{2n}}, \mf{S}\gamma^{E^{2n}})	\leq \epsilon, 
	\end{aeq}%
    where the last inequality is shown in \cref{thm:skeleton_security}.
	
    We now analyze the case of an active adversary as required in the second condition in \cref{def:composable_security}. 
    As each of $\pi^{\auth}_{AP_i}$ and $\pi^{\auth}_{P_iB}$ securely construct authenticated channels $\mbf{A}^+_{AP_i}$, $\mbf{A}^+_{P_iB}$ within $\epsilon'$, respectively, there exist converters $\sigma^E_{AP_i}$ and $\sigma^E_{P_iB}$ such that%
    \begin{aeq}
        d(\pi^{\auth}_{AP_i}(\mbf{K}^{\auth}_{AP_i}\parallel\mbf{C}_{AP_i}), \mbf{A}_{AP_i}^+\sigma^E_{AP_i} ) \leq \epsilon',\\
        d(\pi^{\auth}_{P_iB}(\mbf{K}^{\auth}_{P_iB}\parallel\mbf{C}_{P_iB}), \mbf{A}_{P_iB}^+\sigma^E_{P_iB} ) \leq \epsilon'.
    \end{aeq}%
    We use $\sigma^{E}_{AP_iB}$ to denote  $\sigma^E_{AP_i} \parallel\sigma^E_{P_iB}$ and use $\sigma^{E^n}_{APB}$ to denote $n$-copies of $\sigma^{E}_{AP_iB}$ composed in parallel. For the ease of writing, we use a shorthand notation $\mbf{A}^{+}_{AP_iB}$ to denote $\mbf{A}^{+}_{AP_i} \parallel \mbf{A}^{+}_{P_iB}$ and use $\{\mbf{A}^{+}_{AP_iB}\}$ to denote $\mbf{A}^{+}_{AP_1B} \parallel \dots \parallel \mbf{A}^{+}_{AP_nB}$. Similarly, we use $\pi^{\auth}_{AP_iB}(\mbf{KC})$ to denote $\pi^{\auth}_{AP_i}(\mbf{K}^{\auth}_{AP_i} \parallel \mbf{C}_{AP_i}) \parallel \pi^{\auth}_{P_iB}(\mbf{K}^{\auth}_{P_iB} \parallel \mbf{C}_{P_iB})$. 
  
    From the properties of the pseudo-metric, we have%
    \begin{aeq}\label{eq:auth_2n_channel}
        d(\pi^g\mf{R}_g, \pi^s\mf{R}_s\sigma^{E^n}_{APB}) &\leq d(\{\pi^{\text{auth}}_{AP_iB}(\mbf{KC)}\}, \{\mbf{A}^{+}_{AP_iB}\sigma^E_{AP_iB}\})  \\
        & \leq \sum_{i=1}^n d(\pi^{\text{auth}}_{AP_iB}(\mbf{KC)}, \mbf{A}^{+}_{AP_iB}\sigma^E_{AP_iB}) \\
        & \leq 2n \epsilon',
    \end{aeq}%
    where we use \cref{eq:metric_serial_composition,eq:metric_parallel_composition} for the first inequality to remove common systems in $\pi^g\mf{R}_g$ and $\pi^s\mf{R}_s\sigma^{E^n}_{APB}$, \cref{thm:composability} for the second inequality and the security of the authentication protocol for the third inequality and the fact that there are $2n$ uses of the authentication protocol. 

    From \cref{thm:skeleton_security}, we also have a converter $\sigma^E$ such that%
    \begin{aeq}\label{eq:skeleton_secrecy}
        d(\pi^s\mf{R}_s, \mf{S}\sigma^E ) \leq \epsilon.
    \end{aeq}%
    Thus, we show that for the converter $\sigma'^E = \sigma^E \sigma^{E^n}_{APB}$, we have%
    \begin{aeq}
        &~ d(\pi^g\mf{R}_g, \mf{S}\sigma'^E) \\
        \leq &~ d(\pi^g\mf{R}_g, \pi^s\mf{R}_s\sigma^{E^n}_{APB}) + d(\pi^s\mf{R}_s\sigma^{E^n}_{APB}, \mf{S}\sigma^E\sigma^{E^n}_{APB}) \\
        \leq &~ d(\pi^g\mf{R}_g, \pi^s\mf{R}_s\sigma^{E^n}_{APB}) + d(\pi^s\mf{R}_s, \mf{S}\sigma^E)\\
        \leq &~ 2n\epsilon' + \epsilon,
    \end{aeq}%
    where we use the triangle inequality in the first inequality, \cref{eq:metric_serial_composition,eq:metric_parallel_composition} to drop the common system $\sigma^{E^n}_{APB}$ in the second inequality, and \cref{eq:auth_2n_channel,eq:skeleton_secrecy} in the third inequality.

    Combining those two conditions, we show that $\pi^g$ securely constructs the ideal resource $\mf{S}$ out of the resource $\mf{R}_g$ within $\max(\epsilon, \epsilon + 2n \epsilon')=\epsilon + 2n \epsilon'$.
\end{proof}

We then show the $\epsilon$-robustness of the general DSKE protocol.
\begin{theorem}[Robustness of a general DSKE protocol]\label{thm:robustness_general}
    When the upper bound on the number of compromised Security Hubs is no greater than $\min(n-k, k-1)$, a general DSKE protocol is $\epsilon$-robust with $\epsilon = \min({n \choose k} \frac{m+1}{|F|}, 1)$.
\end{theorem}
\begin{proof}
    When Eve is passive, a general DSKE protocol behaves in the same way as its corresponding skeleton protocol. The reason is as follows. A general DSKE and its corresponding skeleton protocol differ in the availability of authenticated channels. In the general DSKE protocol, an authentication protocol constructs an authenticated channel out of a secret key resource and an insecure channel, while in the skeleton protocol, authenticated channels are assumed to be available. When Eve is passive, she does not tamper each communication channel. Each authenticated channel in the skeleton protocol and each channel constructed by the authentication protocol in the general protocol both faithfully transmit messages. Thus, the result of \cref{thm:robustness_skeleton} directly applies.
\end{proof}
 
\FloatBarrier

\section*{Acknowledgement}
We would like to thank David Jao for helpful comments. This work is supported by MITACS, Natural Sciences and Engineering Research Council of Canada (NSERC), and particularly Innovative Solutions Canada. 
\addcontentsline{toc}{section}{Acknowledgement}

\newpage
\appendices
\crefalias{section}{appendix}

\section{DSKE Protocol} \label{app:protocol}

\begin{table}[!ht]
    \begin{center}
%    \resizebox{0.88\linewidth}{!}{
    \begin{tabular}{|c|l|}
    	\hline
    	\textbf{Symbol} & \textbf{Description} \\
    	\hline
    	$A$ & reference to Alice, or shorthand for $A_i$ \\
    	\hline
    	$A_i$ & Hub $P_i$'s identifier for Alice \\
    	\hline
    	$f$ & polynomial over field $F$ \\
    	\hline
    	$g$	& bijective mapping $\{0,\dots,|F|-1\} \to F$ \\
    	\hline
    	$\mathbf{H}$	& family of 2-universal hash functions \\
    	\hline
        $\mathbf{H}'$	& family of 2-universal hash functions \\
    	\hline
    	$h'_{u^A}$	& function in $\mathbf{H}'$ selected by $u^A$ \\
    	\hline
    	$h_{v^A_i}$ & function in $\mathbf{H}$ selected by $v^A_i$ \\
    	\hline
    	$H^A_i$	 & first table, for Hub $P_i$ and Alice \\
    	\hline
    	$\overline{H}^A_i$ &	second table, for Hub $P_i$ and Alice \\
    	\hline
    	$H^B_i$ &	first table, for Hub $P_i$ and Bob \\
    	\hline
    	$\overline{H}^B_i$ &	second table, for Hub $P_i$ and Bob \\
    	\hline
    	$i$	& $i\in\{1,\dots,n\}$ is a share (Hub) index. \\
    	\hline
    	$j^A_i$	& index into $H^A_i$ \\
    	\hline
    	$\overline{j}^B_i$ & index into $\overline{H}^B_i$ \\
    	\hline
    	$k$	& threshold for the sharing scheme \\
    	\hline
    	$K^A$ 	&	identifier that tracks the current $S^A$ \\
    	\hline
    	$m$	 & number of field elements in $S^A$ \\
    	\hline
    	$M^A_i$	 & message from Alice to Hub $P_i$, no tag \\
    	\hline
    	$\overline{M}^B_i$ &	message from Hub $P_i$ to Bob, no tag \\
    	\hline
    	$n$	& number of Hubs (and hence shares) \\
    	\hline
    	$o^A$ &	secret-authenticating tag for $S^A$ \\
    	\hline
    	$P_i$	& identifier referring to the $i$th Hub \\
    	\hline
    	$R^A_i$	 & $3+m$ field elements of $H^A_i$ \\
    	\hline
    	$\overline{R}^B_i$ & $3+m$ field elements of $\overline{H}^B_i$ \\
    	\hline
    	$S^A$	& Alice's secret of $m$ field elements \\
    	\hline
    	$t^A_i$	& message tag for $M^A_i$ using $v^A_i$ \\
    	\hline
    	$\overline{t}^B_i$	& message tag for $\overline{M}^B_i$ using $\overline{v}^B_i$ \\
    	\hline
    	$u^A$	& $3$ field elements of $Y^A_i$ for $o^A$ \\
    	\hline
    	$v^A_i$	& $2$ field elements of $H^A_i$ for $t^A_i$	 \\
    	\hline
    	$\overline{v}^B_i$	& $2$ field elements of $\overline{H}^B_i$ for $\overline{t}^B_i$ \\
    	\hline
    	$x_i$	& $x$-coordinate for Hub $P_i$ \\
    	\hline
    	$Y^A_i$	& $3+m$ field elements: $u^A\parallel S^A$ \\
    	\hline
    	$Z^A_i$	& encrypted $Y^A_i$ in the message $M^A_i$ \\
    	\hline
    	$\overline{Z}^B_i$	& encrypted $Y^A_i$ in the message $\overline{M}^B_i$ \\
    	\hline
    \end{tabular}
%    }
    \end{center}
    \caption{\label{table:symbols}Symbols used in the protocol description in \Cref{app:protocol}. \textit{Note}: Symbol capitalization may differ from the main text.}
\end{table}

%\FloatBarrier

\subsection{Parameter choices}

We present the general DSKE protocol (as described in the DKSE paper \cite{Lo2022}) with the following simplifying choices in this paper with the aim of establishing a baseline proof of its security.

\begin{itemize}[leftmargin=1.5em]
    \setlength\itemsep{0em}

    \item A predetermined finite field $F$ is used throughout, and \textit{element} refers an element of $F$.
 
    \item The final secret length $m$ is predetermined. 
 
    \item The parameters of the threshold sharing scheme, $n$ and $k$, are predetermined. For clarity, also $k_B = k$ in \cite{Lo2022}.

    \item The number of recipients is limited to 1.
 
    \item The hash function families are predetermined:
    $\mathbf{H} = \{h_{c,d}: F^m \to F: (y_{1},\dots,y_{m}) \mapsto d + \sum_{j=1}^m c^{j} y_{j} \}$ and
    $\mathbf{H}' = \{h'_{c,d,e}: F^m \to F: (y_{(1)},\dots,y_{(m)}) \mapsto d + ce + \sum_{j=1}^m c^{j+1} y_{(j)} \}$.
 
	\item The Shamir $(n, k)$-threshold secret sharing scheme uses \\$f : F \to F : x \mapsto c_{0} + c_{1}x + \cdots + c_{k-1}x^{k-1}$.

	\item The mapping for assigning $x_i$ to the secret and Hubs $\{0, \dots, n\} \to F : i \mapsto x_i$ is predetermined.

    \item The bijective mapping $g:\{0,\dots,|F|-1\} \to F$ is predetermined. This allows encoding of an integer as an element in the protocol.

	\item A Hub sends each client two tables, e.g. $H^{A}_i$ and $\overline{H}^{A}_{i}$. An overline indicates that it for a Hub's sending use.

	\item Mutual identity validation is assumed and the corresponding identifiers $P_{i}, A_{i}, B_i$ are predetermined.

	\item The identifiers $A_{i}$ are equal for all $i$, and we write the identifiers $A_{i}$ as simply $A$ (notation abuse), as for $B$.  $A$ and $B$ also denote the identities Alice and Bob.

	\item The message tag validation key is used once only.

\end{itemize}

\subsection{Baseline protocol}

\begin{enumerate}[label={(\arabic*)}, ref=\arabic*]

	\item \textbf{PSRD generation and distribution}
	
    Honest Hubs securely provide a copy of the two tables of ordered elements (for Alice, $H^A_i$ and $\overline{H}^{A}_{i}$), with assured mutual identity verification, data confidentiality and data authenticity. 

    For simplicity, Alice tracks use of elements in $H^A_i$ by retaining an integer offset $j^A_i$ into $H^A_i$ up to which the elements have been used and erased, initialized as $j^A_i:=0$ upon receiving the tables from $P_i$. Each Hub similarly tracks usage in $\overline{H}^{A}_i$ through $\overline{j}^A_i$.  Since messages may be received out of order, the receiver must individually track which elements of the table have been used.

    \item \textbf{Peer identity establishment}
    
    Alice and Bob need to establish the authenticity of each other's identities. This phase, which in a practical implementation is necessary, is made redundant by the assumed mutual knowledge of identities and identifiers in the form presented in this paper. In practice, each DSKE client can query the identities of other clients with the help of Security Hubs and an information-theoretically secure message tag. See \cite{Lo2022} for details.

    \item \textbf{Secret agreement}

    \begin{enumerate}[label=(\alph*),leftmargin=0.5em, ref=\theenumi{}.\alph*]
        \item \textit{Share generation}
        \begin{enumerate}[label=(\roman*),leftmargin=0.5em, ref=(\theenumii{}.\roman*)]
            \item Alice retrieves the unused sequences $R^A_i$ (length $3+m$) and $v^A_i$ (length $2$) from at offset $j^A_i$ in $H^A_i$, erases them from the table, and remembers $j^A_i+3+m+2$ as $j^A_i$ on the next iteration of the protocol. \label[step]{step:3a1}

            \item Alice uses $3+m$ identical but independent Shamir sharing schemes in parallel, one for each field element in the sequence.
            Alice sets:%
            \begin{aeq}
                Y^A_i := R^A_i \quad  \forall i \in \{1,\dots,k\}.
            \end{aeq}%
            \begin{aeq}
                f_{-2}(x_i)\parallel f_{-1}(x_i)\parallel f_{0}(x_i)\parallel\dots\parallel f_{m}(x_i):=Y^A_i \\ \forall i \in \{1, \dots, k\}.
            \end{aeq}%
            \item Alice solves for $Y^A_i$ in the following, arriving at $n$ shares $Y^A_i$ for $i \in\{1, \dots, n\}$, each being a $(3+m)$-element sequence:%
            \begin{aeq}\label{eq:share}
                f_{-2}(x_i)\parallel f_{-1}(x_i)\parallel f_{0}(x_i)\parallel\dots\parallel f_{m}(x_i)=:Y^A_i \\ \forall i \in \{k+1, \dots, n\}.
            \end{aeq}%
        \end{enumerate}

        \item \textit{Share distribution}
        \begin{enumerate}[label=(\roman*),leftmargin=0.5em]
            \item\label{step:share_distribution_alice} Operations by Alice for share distribution:
            \begin{enumerate}[label=(\arabic*),leftmargin=0.5em]
                \item Alice solves for the secret $Y^A_0=f_{-2}(x_0)\parallel f_{-1}(x_0)\parallel f_{0}(x_0)\parallel\dots\parallel f_{m}(x_0)$, as she did for the shares in \cref{eq:share}.

                \item Alice partitions $Y^A_0$ into $u^A$ of 3 elements and $S^A$ of $m$ elements:%
                \begin{aeq}
                    Y^A_0 =: u^A \parallel S^A.
                \end{aeq}%

                \item Alice calculates an authentication code $o^A$ for the secret to be agreed, which we call the \textit{secret-authenticating tag}:%
                \begin{aeq}
                    o^A := h'_{u^A}(S^A)
                \end{aeq}%
                using the 2-universal family of hash functions $\mathbf{H}'$ with the choice specified by $u^A$.

                \item Alice calculates%
                \begin{aeq}
                    Z^A_i := Y^A_i - R^A_i \quad i \in \{1, \dots, n\}.
                \end{aeq}%
                Note:  $Z^A_i$ is just a tuple of $3+m$ zero elements for each $i \in \{1, \dots, k\}$ because of the cancellation. The protocol could omit this field, but it is kept here to simplify the presentation.

            	\item Alice chooses a secret identification number $K^A$ so that $(A,K^A)$ is unique.
            	The message is:%
                \begin{aeq}
                    M^A_i :=  P_i \parallel A \parallel B \parallel K^A \parallel g(j^A_i) \parallel Z^A_i \parallel o^A.  
                \end{aeq}%
                Fields that omit the index $i$ are the same across all $n$ messages sent to the Hubs, such as $o^A$. $t^A_i$ is different in each message.

                \item Alice calculates the message tag $t^A_i$ using the function family $\mathbf{H}$ as%
                \begin{aeq}
                    t^A_i := h_{v^A_i} ( M^A_i ).
                \end{aeq}%

                \item Alice sends the element sequence%
                \begin{aeq}
                    M^A_i \parallel t^A_i 	
                \end{aeq}%
                to Hub $P_i$, for $i \in \{1, \dots, n\}$.
            \end{enumerate}

            \item Operations by each Hub $P_i$ for share distribution, related to Alice:
            \begin{enumerate}[label=(\arabic*),leftmargin=0.5em]
                \item Hub $P_i$ receives the sequence from Alice,%
                \begin{aeq}
                    M^A_i \parallel t^A_i,
                \end{aeq}%
                which may be corrupted or even lost, which it splits into its components $M^A_i$ and $t^A_i$.

                \item The Hub splits $M^A_i$ into its components via%
                \begin{aeq}
                    M^A_i =:  P_i \parallel A \parallel B \parallel K^A \parallel g(j^A_i) \parallel Z^A_i \parallel o^A.
                \end{aeq}%

            	\item The Hub verifies that the tuple $(P_i, A, B)$ is allowable and was received via the routing from $A$, discarding the message if either check fails. The latter check is significant for robustness against depletion of $H^A_i$.

            	\item The Hub also verifies that the $3+m+2$ elements starting at offset $j^A_i$ in its copy of $H^A_i$ are still unused, discarding the message if any of these elements have been used. Discarding the message at any time up to this point does not deplete elements in the table.
    
            	\item The Hub retrieves $R^A_i$ of $3 + m$ elements at offset $j^A_i$ and $v^A_i$ of 2 elements at offset $j^A_i + 3 + m$ from the table, marking them as used and erasing them from the table.

                \item The Hub verifies the relation%
                \begin{aeq}
                    t^A_i = h_{v^A_i}( M^A_i ), 	
                \end{aeq}%
                discarding the message if this fails. The portion of the table addressed by the message cannot be re-used on failure, due to the single-use constraint in the simplifying assumptions.

                \item The Hub calculates%
                \begin{aeq}
                    Y^A_i := Z^A_i + R^A_i .
                \end{aeq}%
            \end{enumerate}

            \item Operations by each Hub $P_i$ for share distribution, related to Bob:
            \begin{enumerate}[label=(\arabic*),leftmargin=0.5em]
                \item The Hub chooses $\overline{R}^B_i$ and $\overline{v}^B_i$ from $\overline{H}^B_i$ using $\overline{j}^B_i$ in the same way that Alice did from $H^A_i$ in \cref{step:3a1} with the corresponding variables, in the process erasing elements from $\overline{H}^B_i$.  The Hub uses $\overline{R}^B_i$ as an encryption key; unlike Alice, it never treats this as a share.

                \item The Hub calculates%
                \begin{aeq}
                    \overline{Z}^B_i := Y^A_i - \overline{R}^B_i .
                \end{aeq}%

            	\item The Hub generates the message $\overline{M}^B_i$:%
            	\begin{aeq}
                    \overline{M}^B_i :=  P_i \parallel A \parallel B \parallel K^A \parallel g(\overline{j}^B_i) \parallel \overline{Z}^B_i \parallel o^A .
            	\end{aeq}%

                \item The Hub calculates the message tag $\overline{t}^B_i$ as%
                \begin{aeq}
                    \overline{t}^B_i := h_{\overline{v}^B_i} (\overline{M}^B_i) 
                \end{aeq}%

                \item The Hub sends to Bob the element sequence%
                \begin{aeq}
                    \overline{M}^B_i \parallel \overline{t}^B_i .
                \end{aeq}%
            \end{enumerate}

            \item Operations by Bob for each Hub $P_i$ for share distribution, related to Alice:
            \begin{enumerate}[label=(\arabic*),leftmargin=0.5em]
                \item Bob receives the sequence from Hub $P_i$,%
                \begin{aeq}
                    \overline{M}^B_i \parallel \overline{t}^B_i, 
                \end{aeq}%
                which he splits into its components.
    
                \item Bob then splits $\overline{M}^B_i$ into its components%
                \begin{aeq}
                    \overline{M}^B_i =: P_i \parallel A \parallel B \parallel K^A \parallel g(\overline{j}^B_i) \parallel \overline{Z}^B_i \parallel o^A. 
                \end{aeq}%

                \item Bob verifies that the tuple $(P_i, A, B)$ is allowable and was received via the routing from $P_i$, discarding the message if either check fails.  Performing the latter check before proceeding further is significant for robustness.

                \item Bob verifies that the $3+m+2$ elements from offset $\overline{j}^B_i$ in his copy of $\overline{H}^B_i$ are unused, failing which the message is discarded.  Discarding the message at any time up to this point does not deplete elements in the table.

                \item Bob retrieves $\overline{R}^B_i$ of $3 + m$ elements at offset $\overline{j}^B_i$ and $\overline{v}^B$ of $2$ elements at offset $\overline{j}^B_i+3+m$ from the table, marking them as used and erasing them from the table.

                \item Bob then verifies the relation%
                \begin{aeq}
                    \overline{t}^B_i = h_{\overline{v}^B_i} (\overline{M}^B_i), 	
                \end{aeq}%
                discarding the message if this fails.  Note that the portion of the table addressed by the message cannot be re-used on failure, due to the single-use constraint in the simplifying assumptions.

                \item Bob then calculates%
                \begin{aeq}
                    Y^A_i := \overline{Z}^B + \overline{R}^B_i. 
                \end{aeq}%
            \end{enumerate}
        \end{enumerate}

        \item \textit{Secret reconstruction}

        \begin{enumerate}[label=(\roman*),leftmargin=0.5em]
        	\item Bob assembles all subsets of $k$ messages that have $(A, B, K^A, o^A)$ in common.  Associated with each message is a tuple $(x_i, Y^A_i)$.
        	\item Bob solves for a candidate $Y^A_0$ in $f_{-2}(x_0) \parallel \dots \parallel f_m(x_0) = Y^A_0$ from the $(x_i, Y^A_i)$ tuples of each subset, similar to the operation that Alice did in \ref{step:share_distribution_alice}, except with the set of indices $i$ varying by subset, obtaining a candidate per subset. Bob may eliminate duplicates here. 
            \item Bob partitions each distinct candidate $Y^A_0$ of two strings of $3$ and $m$ elements respectively:%
            \begin{aeq}
            	Y^A_0 =: u^A \parallel S^A 	
            \end{aeq}%
            and forms the candidate tuple $(u^A, S^A, o^A)$. 
        \end{enumerate}
    \end{enumerate}

    \item \textbf{Secret validation}

    Bob discards those candidate tuples $(u^A, S^A, o^A)$ for which the following relation does not hold:%
    \begin{aeq}
        o^A = h'_{u^A} ( S^A ). 
    \end{aeq}%
    Bob aborts the protocol if he has no remaining candidate tuples or non-identical remaining candidate tuples. Otherwise, he terminates the protocol with the secret $S^A$. 
    The DSKE protocol ends at this point. The tuple $(A, B, K^A, S^A)$, is known by both Alice and Bob.

\end{enumerate}

\begin{remark}
	Only Bob knows whether the protocol completed successfully. Communicating the completion and the use of the secret can be managed through the tuple $(A, B, K^A)$ and subsequent communication.
\end{remark}

\FloatBarrier

\section{Security of hashing for messages} \label{app:hash_message}

Consider a family of polynomial functions, where $c$, $d$ and $v_{j}$ are elements of a finite field $F$ \cite[Section 4.2]{Bernstein2007}:%
\begin{aeq}
    \mbf{H}=\{h_{c,d}\colon{}F^{s}\to{}F\colon{}(v_{1},\dots,v_{s})\mapsto{}d+\sum_{j=1}^{s}c^{j}v_{j}\}
\end{aeq}%
\Cref{thm:message_correctness} gives the best guessing probability that Eve's message $\mbf{v}^{*}$ is both modified from $\mbf{v}$ and validates against any tag $t^{*}$ with the selection of the hash function $h_{c,d}$ unknown, excluding the case of an empty message ($s=0$).

\begin{theorem}\label{thm:message_correctness}
    Denote $\mbf{v} = (v_{1},\dots,v_{s})$ and $\mbf{v}^{*} = (v^{*}_{1},\dots,v^{*}_{s})$. \\
    Let $\Omega=F^2$ be a sample space with uniform probability. \\
    Let $h_{C,D}(\mbf{v})=D+\sum_{j=1}^{s}C^{j} v_{j}$ define a family of functions with random variables $(C,D)\in\Omega$ as selection parameters. 
    Let $s\ne0$. Let $t\in F$ be given.
    Then, $\max_{t^{*},\mbf{v}^{*}\ne\mbf{v}} \Pr(t^{*}=h_{C,D}(\mbf{v}^{*})~|~t=h_{C,D}(\mbf{v})) = \min(\frac{s}{|F|},1)$.
\end{theorem}
\begin{proof}
    We are given that $t=h_{C,D}(\mbf{v})$. This may be written%
    \begin{aeq}\label{eq:thm_message_correctness_a}
        t=D+\sum_{j=1}^{s}C^{j} v_{j}.
    \end{aeq}%
    This constraint serves to eliminate all pairs $(C,D)$ that do not solve \cref{eq:thm_message_correctness_a}. For every value of $C$ in $F$, this equation determines a unique value for $D$, resulting in exactly $|F|$ pairs that meet the constraint.  Conversely, for every value of $D$, there may be many values of $C$, and by inference, there may be many values of $D$ for which there is no corresponding solution for $C$. Since the \textit{a priori} probability on $\Omega$ (prior to imposing the constraint) is uniform, and each value for $C$ occurs in a pair exactly once, the \textit{a posteriori} marginal distribution on $C$ (i.e. given the constraint) is uniform, but the marginal distribution on $D$ is potentially nonuniform.

    To obtain the probability that $t^{*}=h_{C,D}(\mbf{v}^{*})$ holds, we subtract from it the given \cref{eq:thm_message_correctness_a} to obtain an equivalent equation%
    \begin{aeq}\label{eq:thm_message_correctness_b}
        t^{*}-t=\sum_{j=1}^{s}C^{j} (v^{*}_{j}-v_{j}).
    \end{aeq}%

    For any values of $t$, $t^{*}$, $\mbf{v}$ and $\mbf{v}^{*}\ne\mbf{v}$, the polynomial equation \cref{eq:thm_message_correctness_b} has up to $\min(s,|F|)$ solutions for $C$, with the bound attainable for some $\mbf{v}^{*}$.  Since $C$ is uniform, each solution has probability $\frac{1}{|F|}$.  Thus with $s\ne0$,%
    \begin{aeq}
        \max_{t^{*},\mbf{v}^{*}\ne\mbf{v}} \Pr(t^{*}-t=\sum_{j=1}^{s}C^{j}(v^{*}_{j}-v_{j}))=\min\big(\frac{s}{|F|},1\big)
    \end{aeq}%
    It follows from the equivalence with \cref{eq:thm_message_correctness_a} that%
    \begin{aeq}
        \max_{t^{*},\mbf{v}^{*}\ne\mbf{v}} \Pr(t^{*}=h_{C,D}(\mbf{v}^{*})~|~t=h_{C,D}(\mbf{v}))=\min\big(\frac{s}{|F|},1\big).
    \end{aeq}%
\end{proof}

\FloatBarrier

\section{Security of validated secret sharing} \label{app:hash_secret}

\subsection{Share manipulation in secret sharing}

We use a prime to denote an associated additive difference variable. For example, $a'$ denotes a difference (error), used to produce $a+a'$ from $a$.

\begin{lemma} \label{lem:sum_independence}
    Let $\Omega=F^2$ be a sample space with $(D,X)\in\Omega$. Let $D$ and $X$ be independent random variables with $D$ uniform. Let $z\in F$ be arbitrary. \\
    Then $\Pr(D+X=z)=\frac{1}{|F|}$.
\end{lemma}
\begin{proof}
    For any $z\in F$,%
    \begin{aeq}
        &\Pr(D+X=z) & \\
        &    = \sum_{x\in F}\Pr(D+X=z|X=x)\Pr(X=x) & [\text{by \cref{eq:prob_prop}}] \\
        &    = \sum_{x\in F}\frac{1}{|F|}\Pr(X=x)  & [\text{by \cref{eq:prob_uniform}}] \\
        &    = \frac{1}{|F|}    & [\text{by \cref{eq:prob_sum_to_one}}].
    \end{aeq}%
\end{proof}

\begin{lemma}\label{lem:indepsum}
    Let $\Omega=F^2$ be a sample space. Let random variables $(D,X)\in\Omega$ be mutually independent and $D$ uniform.
    Then $X+D$ and $X$ are mutually independent.
\end{lemma}
\begin{proof}
    For each $(x,d) \in \Omega$,%
    \begin{aeq}
        &\Pr((X+D,X) = (x+d,x)) \\
        = ~& \Pr((D,X)=(d,x)) & [\text{since }X=x] \\ 
        = ~& \Pr(D=d)\Pr(X=x) & [\text{independence of } \\
        & & D\text{ and }X] \\
        = ~& \frac{1}{|F|}\Pr(X=x) & [\text{uniformity of }D] \\
        = ~& \Pr(X+D=x+d)\Pr(X=x) & [\text{by \cref{lem:sum_independence}}].
    \end{aeq}%
    By \cref{eq:prob_independent}, it follows that $X+D$ and $X$ are mutually independent.
\end{proof}

Consider an $(n,k)$-threshold Shamir secret sharing scheme that uses a polynomial over the field $F$.

\begin{theorem}\label{thm:shamir_linearity}
% Note: this commented-out part is an early attempt at a more rigorous statement of the theorem.
%    Denote $Y_U=(Y_{1},\dots,Y_{k})$; $U=[1,\dots,k]$. \\
%    Let each $x_j\in F$ be unique for $j\in U$. \\
%    Let $Y_j\in F$ be mutually independent and uniform. \\
%    Let $\forall j\in U\colon Y_j=\sum_{i=1}^{k}C_i{x_j}^i$ for some $C_i$ (sharing). \\
    In a Shamir secret sharing scheme with threshold $k$, for any given set of $k$ shares, the secret is a linear combination of the shares.
\end{theorem}
\begin{proof}
    A Shamir scheme is based on a polynomial $f$ of degree (at most) $k-1$ over $F$, where $1\le{k}\le{n}<|F|$. In such a scheme, $n+1$ distinct $x$-coordinate values are chosen, with one ($x_{0}$) associated with the secret and the rest ($x_{1},\dots,x_{n}$) each associated with a share.  The secret sharing scheme is defined by the polynomial in $x$ as%
    \begin{aeq}
        f(x)=\sum_{j=0}^{k-1}c_{j}x^{j},
    \end{aeq}%
    with secret coefficients $c_j$.  The $c_j$ and $x_i$ determine the secret $y_0$ and shares $y_{1},\dots,y_{n}$ through $y_i=f(x_i)$.  Any $k$ of the $n+1$ pairs $(x_i, y_i)$ uniquely determine the $c_j$.  The $y_i$ for $k$ of the shares are required to be independent and uniform in $F$.
    For any $J\subseteq\{1,\dots,n\}$ with $|J|=k$, the polynomial $f$ can be expressed as a linear combination of a Lagrange basis of $k$ polynomials $L_j$ of degree $k-1$, such that $L_i(x_j)=\delta_{i,j}$ for $i,j\in{}J$:%
    \begin{aeq}
        f(x)=\sum_{i\in J}y_{i}L_{i}(x).
    \end{aeq}%
    Given known $x_i$ and a set $J$ as defined above, the basis polynomial $L_i$ corresponding to $x_i$ can be determined.  Since each polynomial $L_i$ is of degree $k-1$ and has $k-1$ distinct zeros, $L_i(x_j)\ne0$ whenever $j\not\in{}J$, and in particular, $L_i(x_0)\ne0$.  From this, the secret is a known linear combination of any given $k$ of the $n$ shares:%
    \begin{aeq}
        y_{0}=f(x_{0})=\sum_{i\in J}y_{i}L_i(x_{0}).
    \end{aeq}%
    The confidentiality of $y_0$ provided by the Shamir scheme relies on at least $n-k+1$ of the $n$ shares remaining secret.  Further, if a value $y_i+y'_i$ is substituted for each share $y_i$, the reconstructed secret is%
    \begin{aeq}
        \sum_{i\in J}(y_{i}+y'_{i})L_i(x_{0})=y_{0}+\sum_{i\in J}y'_{i}L_{i}(x_{0}).
    \end{aeq}%
    From this it may be seen that, because the $L_i(x_0)$ are known and nonzero, a single error value $y'_i=y'_0/L_i(x_0)$ added to $y_i$ replaces the reconstructed secret with $y_0+y'_0$.  With multiple errors $y'_i$, $y'_0=\sum_{i\in J}y'_iL_i(x_0)$ can be chosen by an adversary who is able to choose $y'_i$ for a nonempty subset of $J$.
\end{proof}

\subsection{Confidentiality of Shamir sharing scheme}
The fundamental confidentiality of the Shamir secret sharing scheme is expressed in \cref{thm:shamir_confidentiality} \cite{shamir1979share}.
\begin{theorem}\label{thm:shamir_confidentiality}
    In a Shamir secret sharing scheme with threshold $k$, the shared secret is independent of any subset of the shares of size $k-1$ or less.
\end{theorem}
\begin{proof}
    Let $J\subseteq\{1,\dots,n\}$, with $|J|=k$, and $i\in J$.
    From \cref{thm:shamir_linearity}, the secret $Y_0$ is a linear sum of $k$ shares:%
    \begin{aeq}
        Y_0=\sum_{j\in J}d_{j}Y_{j},
    \end{aeq}%
    where each coefficient $d_{j}=L_{j}(x_{0})$ is nonzero for all $j\in J$. In a Shamir sharing scheme each subset of $k-1$ shares is mutually independent and uniform. The secret $Y_{0}$ can therefore be partitioned into a sum for any $i\in{}J$:%
    \begin{aeq}
        Y_{0}=d_{i}Y_{i}+\sum_{j\in J\setminus\{i\}}d_{j}Y_{j}.
    \end{aeq}%
    Define $D=d_{i}Y_{i}$ and $X=\sum_{j\in J\setminus\{i\}}d_{j}Y_{j}$. Since $d_{i}$ is a nonzero field element and $Y_{i}$ is uniform and independent of $X$, $D$ and $X$ are mutually independent and uniform.  By \Cref{lem:indepsum}, $Y_0=D+X$ and $X$ are mutually independent, and thus $Y_0$ is independent of any subset that excludes share $i$, for any $i\in J$.

    Note that any subset with size less than $k-1$ is included in the sum $\sum_{j\in J\setminus\{i\}}d_{j}Y_{j}$.  The independence and uniformity of the term $d_{i}Y_{i}$ remains sufficient.
\end{proof}

\subsection{Polynomial hash function}

In this subsection, we use the following notations:
\begin{itemize}[leftmargin=0.5cm]
    \setlength\itemsep{0em}
    \item $y_{(j)}$ denotes the secret ($y_{0}$) from the $j$th of $3+m$ secret sharing schemes being run in parallel to build a $(3+m)$-element secret $Y^A_{0}$, which $j\in\{-2,-1,0,\dots,m\}$ indexes into. Set $c:=y_{(-2)}$, $d:=y_{(-1)}$ and $e:=y_{(0)}$.
    \item The differences $c'$, $d'$, $e'$ and $y'_{(j)}$ each correspond to $c$, $d$, $e$ and $y_{(j)}$ as the difference $y'_0$ for the $y_0$ in the proof of \cref{thm:shamir_linearity}.
\end{itemize}
Consider the family of functions defined by%
\begin{aeq}\label{eq:hash_family_2}
    \mbf{H}'=\{h'_{c,d,e}\colon& F^m\to F\colon\\
    &(y_{(1)},\dots,y_{(m)})\mapsto{}d+ce+\sum_{j=1}^{m}c^{j+1}y_{(j)}\},
\end{aeq}%
where $c$, $d$ and $e$ are elements of a finite field $F$.
Note that this is essentially the same family of functions as used for message hashing: the first element in the list has been shown as a subscript due to its role as a function-selection parameter, and a prime has been added to distinguish it. 

\begin{theorem}\label{thm:shamir_correctness}
    Denote $\mbf{y}=(y_{(1)},\dots,y_{(m)})$ and $\mbf{y}'=(y'_{(1)},\dots,y'_{(m)})$.
    Let $\Omega=F^3$ be a sample space with uniform probability.
    Let $h'_{C,D,E}(\mbf{y})=D+CE+\sum_{j=1}^{m}C^{j+1}y_{(j)}$ define a family of hash functions with random variables $(C,D,E)\in\Omega$ as selection parameters. Let $m\ne0.$
    Then, $\max_{t',c',d',e',\mbf{y}'\ne0} \Pr(t+t'=h'_{C+c',D+d',E+e'}(\mbf{y}+\mbf{y}')~|~t=h'_{C,D,E}(\mbf{y}))\le\min(\frac{m+1}{|F|},1)$.
\end{theorem}
\begin{proof}
    Given a Shamir secret sharing scheme used to transmit all of $C$, $D$, $E$ and $\mbf{y}$, an adversary who controls from $1$ to $k-1$ shares can modify $C$, $D$ and $E$ simultaneously by adding a chosen constant to each (as per \cref{thm:shamir_linearity}), with $t$ and $\mbf{y}$ assumed known and modifiable.  Thus, $t$, $C$, $D$, $E$ and $\mbf{y}$ are replaced with $t+t'$, $C+c'$, $D+d'$, $E+e'$ and $\mbf{y}+\mbf{y}'$ respectively.

    We are given that $t=h'_{C,D,E}(\mbf{y})$, which may be written%
    \begin{aeq}\label{eq:hash_shamir_given}
        t=D+CE+\sum_{j=1}^{m}C^{j+1}y_{(j)}.
    \end{aeq}%
    This determines a unique value $D$ for every pair of values $C,E$, and since the \textit{a priori} probability on $\Omega$ is uniform, the \textit{a posteriori} marginal distribution over pairs $(C,E)$ remains uniform, but the marginal distribution on $D$ is potentially nonuniform by a similar argument as in the proof of \cref{thm:message_correctness}.
    To obtain the probability that $t+t'=h'_{C+c',D+d',E+e'}(\mbf{y}+\mbf{y}')$ holds, we subtract the given \cref{eq:hash_shamir_given} from it to obtain the equivalent equation%
    \begin{aeq}\label{eq:hash_shamir_attack}
        t'=~&d'+c'E+Ce'+c'e'\\
        &+\sum_{j=1}^{m}[(C+c')^{j+1}(y_{(j)}+y'_{(j)})-C^{j+1}y_{(j)}].
    \end{aeq}%
    When $c'=0$ and $\mbf{y}'\ne0$, \cref{eq:hash_shamir_attack} reduces to%
    \begin{aeq}\label{eq:hash_shamir_reduced}
        t'=d'+Ce'+\sum_{j=1}^{m}C^{j+1}y'_{(j)},
    \end{aeq}%
    which is a non-constant polynomial in $C$ since $\mbf{y}'\ne0$ implies that for at least one value of $j$, $y'_{(j)}\ne0$, but is independent of $E$.  By the uniformity of $C$, each distinct root for the polynomial in $C$ has probability $\frac{1}{|F|}$. The number of distinct roots for a polynomial of degree $m+1$ is at most $\min(m+1,|F|)$, giving a maximum probability for \cref{eq:hash_shamir_reduced} holding of $\min(\frac{m+1}{|F|},1)$.

    When $c'\ne0$, \cref{eq:hash_shamir_attack} may reduce either to a non-constant or to a constant polynomial in $C$, but either way it retains the term that depends on $E$.  For the former (where there is a dependency on $C$), for each value of $E$ there may be up to $\min(m+1,|F|)$ values of $C$ that solve the polynomial, as for the case where $c'=0$, again giving a probability of holding of $\min(\frac{m+1}{|F|},1)$.  For the latter (where there is no dependency on $C$), \cref{eq:hash_shamir_attack} reduces to%
    \begin{aeq}\label{eq:hash_shamir_lincase}
        t'=~d'+c'E+c'e',
    \end{aeq}%
    which, by the uniformity of $E$ and that $c'\ne0$, has probability $\frac{1}{|F|}$ of holding.

    Given that $m$ cannot be negative, $\min(\frac{m+1}{|F|},1)\ge\frac{1}{|F|}$.  Thus, considering all the cases above, the maximum probability of \cref{eq:hash_shamir_attack} holding under the condition $\mbf{y}'\ne0$ is upper-bounded by $\min(\frac{m+1}{|F|},1)$.

    We exclude the case $m=0$, since the $\max$ operator over an empty domain is undefined.

    Thus, we have that for $m\ne0$,%
    \begin{aeq}
        \max_{t',c',d',e',\mbf{y}'\ne0} \Pr(t+t'=h'_{C+c',D+d',E+e'}(\mbf{y}+\mbf{y}')) \\
        \le\min\big(\frac{m+1}{|F|},1\big).
    \end{aeq}%
\end{proof}

\FloatBarrier

\bibliography{references.bib}
\bibliographystyle{IEEEtran}
%------------------------------------------------------------------------------------------------------------------------------------

\end{document}